\pdfoutput=1
\documentclass[12pt]{article}

\usepackage{amscd}
\usepackage{amsfonts}
\usepackage{amsmath}
\usepackage{amssymb}
\usepackage{graphicx}
\usepackage{latexsym}
\usepackage{mathrsfs}
\usepackage{theorem}
\usepackage{color}

\newcommand{\lamia}[1]{{\color{red} #1 \color{black}}}

\newcommand{\g}{\mathfrak{g}}
\newcommand{\p}{\mathfrak{p}}

\newcommand{\R}{{\mathbb{R}}}
\renewcommand{\H}{{\mathbb{H}}}

\newcommand{\C}{{\mathbb{C}}}
\newcommand{\I}{{\mathbb{I}}}

\newcommand{\CP}{{\mathbb{C}}{{P}}}

\newcommand{\beq}{\begin{equation}}
\newcommand{\eeq}{\end{equation}}
\newcommand{\bea}{\begin{eqnarray}}
\newcommand{\eea}{\end{eqnarray}}
\newcommand{\ben}{\begin{eqnarray*}}
\newcommand{\een}{\end{eqnarray*}}
\newcommand{\bem}{\begin{enumerate}}
\newcommand{\eem}{\end{enumerate}}

\newcommand{\ra}{\rightarrow}

\newcommand{\hra}{\hookrightarrow}

\newcommand{\cd}{\partial}
\newcommand{\wt}{\widetilde}

\newcommand{\less}{\backslash}

\newcommand{\mmm}{\wt{M}_2}

\newcommand{\rat}{{\sf Rat}}

\newcommand{\so}{{\mathfrak{so}}}

\def \d{\mathrm{d}}

\newcommand{\ignore}[1]{}

\renewcommand{\aa}{\mathscr{A}}

\renewcommand{\aa}{{\cal A}}

\newcommand{\ol}{\overline}

\newcommand{\sigvec}{\mbox{\boldmath{$\sigma$}}}
\newcommand{\lamvec}{\mbox{\boldmath{$\lambda$}}}

\newcommand{\Omegavec}{\mbox{\boldmath{$\Omega$}}}

\newcommand{\Evec}{\mbox{\boldmath{$E$}}}

\newcommand{\zerovec}{\mbox{\boldmath{$0$}}}

\newcommand{\eps}{\varepsilon}

\theoremstyle{plain}
\newtheorem{thm}{Theorem}

\newtheorem{lemma}[thm]{Lemma}
\newtheorem{prop}[thm]{Proposition}
\newtheorem{proposition}[thm]{Proposition}
\newtheorem{cor}[thm]{Corollary}
\newtheorem{corollary}[thm]{Corollary}

{\theorembodyfont{\rmfamily}

\newtheorem{remark}[thm]{Remark}

}

\newcommand{\news}{\setcounter{equation}{0}}
\newenvironment{proof}{\noindent{\it Proof:\, }}{\hfill$\Box$\vspace*{0.5cm}
}

\renewcommand{\theequation}{\thesection.\arabic{equation}}

\begin{document}

\title{Ricci magnetic geodesic motion of vortices and lumps}
\author{
L.S. Alqahtani\thanks{\tt lalqahtani@kau.edu.sa} \\
Department of  Mathematics, King Abdulaziz University\\
Jeddah, KSA
\and
J.M. Speight\thanks{\tt speight@maths.leeds.ac.uk}\\
School of Mathematics, University of Leeds\\
Leeds LS2 9JT, UK
}

\maketitle

\begin{abstract}
Ricci magnetic geodesic (RMG) motion in a k\"ahler manifold is the analogue of geodesic motion in the presence
of a magnetic field proportional to the ricci form. It has been conjectured to model low-energy dynamics of
vortex solitons in the presence of a Chern-Simons term, the k\"ahler manifold in question being the $n$-vortex
moduli space. This paper presents a detailed study of RMG motion in soliton moduli spaces, focusing on the cases of
hyperbolic vortices and spherical $\CP^1$ lumps. It is shown that RMG flow localizes on fixed point sets of groups
of holomorphic isometries, but that the flow on such submanifolds does not, in general, coincide with their intrinsic
RMG flow. For planar vortices, it is shown that RMG flow differs from an earlier reduced dynamics proposed by Kim and Lee,
and that the latter flow is ill-defined on the vortex coincidence set. An explicit formula for the metric on the whole
moduli space of hyperbolic two-vortices is computed (extending an old result of Strachan's), and RMG motion of centred
two-vortices is studied in detail. Turning to lumps, the moduli space of static $n$-lumps is $\rat_n$, the space of degree $n$
rational maps, which is known to be k\"ahler and geodesically incomplete. It is proved that $\rat_1$ is, somewhat
surprisingly, RMG complete (meaning that that the initial value problem for RMG motion has a global solution for all
initial data). It is also proved that the submanifold of rotationally equivariant $n$-lumps, $\rat_n^{eq}$, a topologically
cylindrical surface of revolution, is intrinsically RMG incomplete for $n=2$ and all $n\geq 5$, but 
that the extrinsic RMG flow
on $\rat_2^{eq}$ (defined by the inclusion $\rat_2^{eq}\hra\rat_2$) is complete.
\end{abstract}

\maketitle

\section{Introduction}
\label{sec:intro}

Let $(M,g,J)$ be a k\"ahler manifold with Ricci form $\rho$, that is,
$\rho (X,Y)=\text{Ric}(JX,Y)$ where $\text{Ric}$ denotes the Ricci tensor defined by $g$. A smooth curve
$\alpha:I\ra M$ is {\em Ricci magnetic geodesic} if
\begin{equation}\label{rmg}
\nabla^\alpha _{d/dt}\dot{\alpha }= \lambda \;\sharp \iota _{\dot{\alpha }} \rho ,
\end{equation}
where $\nabla^\alpha$ is the pullback of the Levi-Civita connexion on $M$ to
$\alpha^{-1}TM$, $\sharp:T^*M\ra TM$
denotes the metric isomorphism, 
$\iota$ denotes interior product, and $\lambda\in\R$ is
a constant parameter. 
We shall call such a curve RMG, or RMG${}_\lambda$ if we wish to
emphasize the role of the parameter $\lambda$.
This is an example of magnetic geodesic flow, that is,
motion of a charged particle, of electric charge $\lambda$, under the influence
of a magnetic field, in this case, the two-form $\rho$. 
Note that the flow
reduces to conventional geodesic motion if $\lambda=0$, and that, in all cases,
RMG curves have constant speed, since
\beq
\frac{d}{dt} \|\dot{\alpha }(t)\|^2
=2 g(\dot{\alpha }(t),\lambda \sharp\, \iota _{\dot{\alpha }(t)}\rho)
=2 \lambda \rho (\dot{\alpha }(t),\dot{\alpha }(t))=0.\label{rmg3}
\eeq
Unlike geodesics, RMG curves depend on the length, and not just the direction,
of their initital velocity.
Clearly, $\alpha(t)$ is RMG${}_\lambda$ if and only if $\wt\alpha(t)=
\alpha(\lambda_*t)$ is RMG${}_{\lambda_*\lambda}$, so we may, without loss of 
generality, scale $\lambda$ to any convenient value, or leave $\lambda$ general
and consider only RMG curves of unit speed. 

RMG flow was first proposed by Collie and Tong \cite{colton}
as a model of the low-energy dynamics of vortex 
solitons in a certain Chern-Simons variant \cite{leeleemin}
of the abelian Higgs model on
$\R^2$. In this setting, $M=M_n\equiv\C^n$, 
the moduli space of static $n$-vortex
solutions of the (usual) abelian Higgs model, and $g$ is its $L^2$ metric.
In the limit $\lambda\ra 0$, one recovers geodesic motion on $M_n$, a 
well-studied problem \cite{sam} which is rigorously known to approximate
low-energy vortex dynamics in the absence of a Chern-Simons term
\cite{stu-vortex}. RMG flow may thus be regarded as a geometrically natural
perturbation of the geodesic approximation of Manton \cite{mansut}, arising
from the inclusion of a Chern-Simons term. 
Low-energy vortex dynamics in this system was previously studied by
Kim and Lee \cite{kimlee}, who, by
a direct perturbative calculation, derived a structurally similar magnetic
geodesic flow on $M_n$. Indeed, Collie and Tong assert \cite{colton}
that the Kim-Lee flow actually {\em is} RMG flow, and that their own
contribution is to generalize and
give it both a geometric interpretation, and an 
alternative (rather indirect) derivation. In fact, we will see
that the Kim-Lee flow on $M_n$ is {\em not} RMG flow, as claimed in 
\cite{colton}, and, further, is not a well-defined flow on $M_n$ at all,
since it is singular on the vortex coincidence set.

In this paper we present a detailed study of RMG flow on the moduli spaces
of abelian Higgs 
vortices and $\CP^1$ lumps. For vortices on $\R^2$, we show that
the Kim-Lee flow is ill-defined on the subset of $M_n$ where two or more
vortices coincide, and hence that this flow cannot coincide with RMG flow 
which is, perforce, globally well-defined. We then consider the model
on the hyperbolic plane of critical curvature, where the vortex equations
are integrable \cite{wit}, and exact $n$-vortex solutions can be written
down. By a careful analysis of the isometric action of $SL(2,\R)$ on
$M_2$, we find an exact formula for its $L^2$ metric, generalizing 
results of Strachan \cite{str}, who computed the induced metric on 
two different two-dimensional submanifolds of $M_2$. We then study
RMG flow on the submanifold of centred two-vortices $M_2^0$
in detail, showing that,
contrary to a claim of one of us in \cite{kruspe-vortex}, this does
{\em not} coincide with the intrinsic RMG flow on $M_2^0$ -- in fact, the two
flows exhibit qualitative differences.

We go on to study RMG flow on $\rat_n$, the space of degree $n$ holomorphic
maps $S^2\ra S^2$ (or, equivalently, the moduli space of $n$ $\CP^1$ lumps on
$S^2$) equipped with its $L^2$ metric. This geometry arises as the
infinite electric charge limit of a certain semi-local vortex model \cite{bap-l2,liu},
so the RMG flow may be relevant to the low energy dynamics of such 
vortices in the presence of a Chern-Simons term. However, our main interest
in it concerns the question of completeness.

Since RMG flow proceeds with constant speed, it is immediate that RMG flow
on any geodesically (or, equivalently, metrically) complete k\"ahler manifold
is {\em complete}, that is, given any initial data $x\in M$, $v\in T_xM$, 
there is a corresponding RMG curve 
$\alpha:\R\ra M$ (well-defined for all times $t\in\R$) with 
$\alpha(0)=x$, $\dot\alpha(0)=v$. The converse
question is nontrivial, however. If a k\"ahler manifold is RMG complete,
does it follow that it is geodesically complete? The time-scaling
properties of RMG flow noted above led one of us to conjecture, in 
\cite{kruspe-vortex}, that the answer is yes: if $M$ is
RMG complete then all RMG${}_\lambda$ curves exist for all time and all 
$\lambda$, and RMG flow tends to geodesic flow as $\lambda\ra 0$ (or, 
equivalently, as speed tends to infinity), so it seems plausible that geodesics 
should likewise exist for all time. 
In fact, this conjecture is false, and $\rat_1$ provides a counterexample:
it is known \cite{spe-l2} to
be k\"ahler and geodesically incomplete but, as will be shown, is RMG 
complete. The point is that the Ricci curvature of $\rat_1$ grows unbounded
as one approaches its boundary at infinity so, even though this boundary
lies at finite distance, an unbounded ``magnetic field'' deflects any 
``charged'' particle from hitting it in finite time. We conjecture
that $\rat_n$ is RMG complete for all $n\geq 2$ also, despite being
geodesically incomplete \cite{sadspe}, and present some evidence in 
favour of this conjecture. 

The rest of this paper is structured as follows. In section \ref{sec:rmg}
we present some generalities on RMG flow on k\"ahler manifolds, including a
useful symmetry reduction lemma. In section \ref{sec:vortex} RMG flow
on vortex moduli spaces is studied, first for vortices on $\R^2$, then
on the hyperbolic plane. In section \ref{sec:lump} RMG flow on $\rat_n$ is
studied, focusing on $\rat_1$. Finally, section \ref{sec:conc} presents
some concluding remarks.

\section{RMG flow}
\label{sec:rmg}
\news

We have already noted that RMG flow (like any magnetic geodesic flow)
conserves speed. Since the Ricci form of a k\"ahler manifold is closed,
one can locally express $\rho$ as $\d \aa$, for some locally defined one-form
$\aa$ on $M$. Then RMG flow has a local Lagrangian formulation,
\beq
L=\frac12\|\dot\alpha(t)\|^2-\lambda \aa(\dot\alpha(t)),
\eeq
that is, $\alpha:[a,b]\ra M$ is RMG if and only if it locally
extermizes $S=\int_a^b Ldt$ among all paths with fixed endpoints.
If $H^2(M)=0$, as in all cases of interest in this paper, this formulation is
actually global. We shall use this fact repeatedly.

\ignore{
In complete generality
(that is, whether or not $\rho$ is exact), RMG flow has a global
hamiltonian formulation: it is flow in $T^*M$ along the symplectic gradient
of
\beq
H=\frac1\|\flat\dot\alpha(t)\|^2
\eeq
with respect to the deformed symplectic form
\beq
\omega_\lambda=\omega_{can}-\lambda\pi^*\omega
\eeq
where $\omega_{can}$ is the canonical symplectic form on $T^*M$, 
$\flat:TM\ra T^*M$ is the metric isomorphism induced by $g$ (that is, 
$\flat X(Y)=g(X,Y)$), $\sharp=\flat^{-1}$  and 
$\omega(\cdot,\cdot)=g(J\cdot,\cdot)$
is the k\"ahler form on $M$. This formulation is a convenient 
starting point for quantizing RMG flow, but will not be used further here. 
}

Unlike geodesics, RMG curves are not invariant under time reversal, and
local isometries
do not necessarily map RMG curves to RMG curves.
However, {\em holomorphic} local isometries do preserve RMG curves:

\begin{proposition}\label{totrmg1}
Let $\varphi  :M\rightarrow N$ be a holomorphic
local isometry between two K\"{a}hler manifolds $M$ and $N$ and 
$\alpha:I\ra M$ be an RMG curve on $M$. Then, $\varphi\circ\alpha$ 
is an RMG curve on $N$.
\end{proposition}

\begin{proof}
Let $\nabla $ and $\bar{\nabla }$ be the Levi-Civita connexions with respect to the k\"{a}hler metrics $g_M$ and $g_N$ on $M$ and $N$, respectively. 
Similarly, denote by $\rho_M,\rho_N$ and $J_M,J_N$ the Ricci forms and
almost complex structures on $M,N$.
Let $\alpha:I\ra M$ be an RMG  curve and $\wt\alpha=\varphi\circ\alpha:I\ra N$. 
Since $\varphi:M\rightarrow N $ is an 
isometry, then \cite{Oneill}
 \begin{equation}\label{iso1}
\d\varphi(\nabla^\alpha_{d/dt}\dot\alpha)=\bar\nabla^{\wt\alpha}_{d/dt}\dot{\wt\alpha}.
\end{equation}
Hence, for any $X\in \Gamma (TM)$,
\bea
g_N(\bar{\nabla }_{d/dt}^{\wt\alpha }\dot{\wt\alpha }, d\varphi  X)
&=&g_N(d\varphi \bigl (\nabla _{d/dt}^{\alpha } \dot{\alpha }\bigr),d\varphi  X)
=g_M(\nabla _{d/dt}^{\alpha }\dot{\alpha}, X)\nonumber \\
&=&g_M(\lambda \sharp_M \iota _{\dot{\alpha }}\rho_M ,X)\qquad\mbox{since $\alpha$ is RMG}\nonumber \\
&=&\lambda \rho_M (\dot{\alpha },X)
=\lambda \text{Ric}_M(J_M\dot{\alpha },X)\nonumber \\
&=&\lambda \text{Ric}_N(\d\varphi J_M\dot\alpha,\d\varphi X)\qquad\mbox{since $\varphi$
is an isometry}\nonumber \\
&=&\lambda \text{Ric}_N(J_N\d\varphi\dot\alpha,\d\varphi X)\qquad\mbox{since $\varphi$
is holomorphic}\nonumber \\
&=& \lambda \rho_N(\dot{\wt\alpha},\d\varphi X)=
g_N(\lambda\sharp\iota_{\dot{\wt{\alpha}}}\rho_N,\d\varphi X).
\eea
But $g_N$ is nondegenerate and $\d\varphi$ surjective, so $\wt\alpha$ is
RMG.
\end{proof}


\begin{corollary}\label{totrmg3}\label{hatu}
Let $M$ be a connected component of a fixed point set of a group of holomorphic isometries of a k\"{a}hler manifold $\overline{M}$. Then, any RMG curve  $\alpha$ on $\overline{M}$  with initial data $\dot{\alpha }(0)\in T_{\alpha (0)}M$ remains on $M$.
\end{corollary}

\begin{proof}
Let $G$ be a group of holomorphic isometries from $\overline{M}$ to itself and let  $M$ be  a connected component of the fixed point set of  $G$.  Let also 
\begin{equation}\label{tot1}
 V_{p}=\{u\in T_{p}\overline{M}: d\varphi  _{p}u=u,\; \forall \varphi  \in G\},\qquad \forall \;p\in M.
\end{equation}
We know that  $M$ is a totally geodesic submanifold of $\overline{M}$ and $T_{p}M=V_{p}$ for all  $p\in M$ \cite{Berndt}.  Now, let $\alpha:I\rightarrow M$ be an  RMG curve on $\overline{M}$ with initial data
\begin{equation}\label{datapv}
\alpha (0)=p\in M,\qquad \dot{\alpha }(0)=v\in T_pM.
\end{equation}
By Proposition \ref{totrmg1}, for  all $\varphi \in G$, the curve   $(\varphi \circ \alpha )(t)$ is RMG on $\overline{M}$. But its initial data are 
\begin{equation}
(\varphi \circ \alpha )(0)=\varphi (p)=p\in M,\qquad (d\varphi \dot{\alpha })(0)=d\varphi _p(v)=v\in T_pM.
\end{equation}
Thus, both $\alpha (t)$ and  $(\varphi \circ \alpha )(t)$ satisfy the RMG equation on $\overline{M}$ with the same initial data, and so  by standard existence and uniqueness theory for ODEs, 
\begin{equation}\label{rmgrmg}
(\varphi \circ \alpha )(t)=\alpha (t),\qquad \forall \;\varphi \in G.
\end{equation}
Hence, $\alpha (t)\in M$  for all time.
\end{proof}

\begin{remark}\label{totrmg4}
One can see that the connected component $M$ of the fixed point set of a group $G$ of holomorphic isometries on a 
k\"{a}hler manifold $\overline{M}$ is a complex submanifold of $\overline{M}$, and so is k\"{a}hler. This follows since
 for all $u\in T_p M=V_p$ and all $\varphi \in G$, 
\begin{equation}\label{tot2}
J_{p}u=J_{p}\:(d\varphi  _{p}u)=d\varphi  _{p}\:(J_{p}u),
\end{equation}
so,  $J_{p}u\in V_p= T_{p}M$ for all $u\in T_{p}M$. It follows that there are two different RMG flows on $M$: the 
original RMG flow on $\ol{M}\supset M$, which preserves $M$, and the RMG flow on $M$ defined by its own
Ricci form $\rho_M$. We shall call these the {\em extrinsic} and {\em intrinsic} RMG flows on $M$ respectively.
Since $\iota^*\rho_{\ol{M}}\neq \rho_M$ in general (where $\iota:M\ra \ol{M}$ denotes inclusion), these two flows on
$M$ do {\em not} coincide in general. 
\end{remark}

\begin{remark}\label{jama}
 For two-dimensional k\"{a}hler manifolds, the RMG equation (\ref{rmg}) simplifies to  
\begin{equation}\label{rmgdim2}
\nabla ^\alpha _{d/dt}\;\dot{\alpha }=\lambda \,\frac{ S }{2}\;J \dot{\alpha },
\end{equation}
where  $S$ denotes  the scalar curvature of $M$. Choosing $\|\dot\alpha\|=1$ for convenience, one
sees that RMG${}_\lambda$ curves are precisely those curves whose geodesic curvature is $\lambda$ times the
Gauss curvature of $M$.
\end{remark}


\section{RMG motion of vortices}
\label{sec:vortex}
\news

The field theory of interest is defined on spacetime $\R^3$ given a Lorentzian metric 
$\eta=dt^2-\Omega(x,y)^2(dx^2+dy^2)$. The conformal factor $\Omega$ will later be chosen so that the spacelike
slice $t=0$ is either the euclidean plane or the hyperbolic plane, but it is convenient to leave it arbitrary at 
first. The theory has, like the abelian Higgs model,
a complex scalar field $\phi$ minimally coupled to a $U(1)$ gauge connexion $A=A_\mu dx^\mu$. It has, in addition, a neutral
(real) scalar field $N$. Its lagrangian density is
\begin{align}
{\cal L}=\frac{1}{2} & \biggl( D_\mu \phi  \overline{D^\mu \phi} -\frac{1}{2} F_{\mu \nu } F^{\mu \nu }+\kappa\;  \epsilon ^{\mu \nu \rho } A_\mu  \partial _\nu A_\rho \nonumber 
+ \partial _\mu N\partial ^\mu N \\
&- \frac{1 }{4} (\left| \phi \right| ^2-1-2\kappa  N)^2+ \left| \phi \right| ^2 N^2 \biggr)\label{csL} 
\end{align}
where $D\phi=\d\phi-iA\phi$, $F=\d A$ and $\kappa$ is a real parameter (the Chern-Simons constant)
which, at the cost of the redefinitions
$t\mapsto -t$, $N\mapsto-N$ if necessary, we may assume is non-negative. 

In order to have finite total energy
\bea
E&=&\frac12\int\bigg(\Omega^2|D_0\phi|^2+|D_1\phi|^2+|D_2\phi|^2+\Omega^{-2}F_{12}^2+F_{0i}F_{0i}
+\Omega^2\cd_0N^2
+\cd_iN\cd_iN \nonumber \\
&&\qquad\qquad\qquad+\frac{1 }{4} \Omega^2(\left| \phi \right| ^2-1-2\kappa  N)^2- \Omega^2\left| \phi \right| ^2 N^2\bigg) dx dy
\eea
the fields $\phi,N$ must have boundary behaviour $|\phi|\ra 1$, $N\ra 0$, or
$\phi\ra 0$, $N\ra-(2\kappa)^{-1}$ as $r=\sqrt{x^2+y^2}\ra\infty$. 
We choose the first possibility, as this allows vortex solutions. Then,
as usual \cite{mansut}, the Higgs field at spatial infinity winds some integer $n$ times around the unit circle
in $\C$, and the total magnetic flux of the field is quantized
\beq
\int Bdxdy=2\pi n
\eeq
where the magnetic field is $B=F_{12}$. There is a Bogomoln'yi
argument \cite{leeleemin} which shows that among all stationary fields (meaning $\cd_0N=0$, $\cd_0\phi=0$) of winding 
$n$, 
\beq
E\geq \pi n
\eeq
with equality if and only if
\begin{align}
(D_1 \pm  iD_2) \phi &=0\qquad & B\pm  \frac{\Omega ^2}{2}(\left| \phi \right| ^2-1-2\kappa N)&=0\label{bog12}\\
A_0\mp   N&=0\qquad & \cd_i\cd_i A_0 -\Omega ^2\left| \phi \right| ^2 A_0- \kappa  B&=0,\label{bog34}
\end{align}
where the upper (lower) signs apply if $n$ is positive (negative). 
A formal index calculation indicates the space of (gauge equivalence classes of) winding $n$ solutions of 
(\ref{bog12}),(\ref{bog34}) has real dimension $2n$ \cite{leeminrim}.

The top two equations (\ref{bog12})
reduce to the usual Bogomol'nyi equations for vortices when $\kappa=0$, and in this case the bottom two
equations (\ref{bog34}) are trivially satisfied by $A_0=N=0$. It follows that, when $\kappa=0$, the moduli space of 
winding $n$ solutions of (\ref{bog12}),(\ref{bog34}) is precisely $M_n$, the space of abelian Higgs $n$-vortices
\cite{tau}. Recall that such vortices are in one-to-one correspondence with unordered $n$-tuples of points in
$\R^2=\C$, the zeros, with multiplicity, of the field $\phi$,
and that their low-energy dynamics is governed by geodesic motion in $M_n$ with respect to $\gamma_{L^2}$,
the $L^2$
metric \cite{str,sam}. There is a useful semi-explicit formula for this metric due to Strachan \cite{str} (on
the hyperbolic plane) and Samols \cite{sam} (on the euclidean plane). Let $\Delta_n$ denote the subset of
$M_n$ on which two or more vortex positions (zeros of $\phi$) coincide. Then on $M_n\less\Delta_n$ we may use 
the zeroes of $\phi$, $(z_1,z_2,\ldots,z_n)\in\C^n$ as local complex coordinates for $M_n$. For a fixed set of vortex positions,
we can expand $\log|\phi(z)|^2$ in a neighbourhood of each $z_r$, $r=1,\ldots,n$,
\beq
\log|\phi(z)|^2=\log|z-z_r|^2+a_r+\frac12\left\{b_r(z-z_r)+\ol{b}_r(\ol{z}-\ol{z}_r)\right\}+\cdots
\eeq
where $b_r$ are $n$ unknown complex functions of $(z_1,\ldots,z_n)$, and $a_r$ are, similarly, unknown real functions. 
Then the metric on $M_n\less\Delta_n$ is
\begin{equation}\label{hy12}
\gamma _{L^2}= \pi \sum_{r,s=1}^n \biggl(\Omega ^2  \delta _{rs}+2 \frac{\partial b_s}{\partial z_r}\biggr)\;dz_r\,d\bar{z}_s.
\end{equation}

Following Kim and Lee \cite{kimlee} and Collie and Tong \cite{colton} we assume that, 
for $\kappa>0$ but small, $n$-vortex solutions of (\ref{bog12}),(\ref{bog34}) remain in bijective correspondence with
unordered $n$-tuples of points in $\C$, and hence with points in $M_n$, and that their low energy dynamics is
described by some perturbed geodesic motion in $(M_n,\gamma_{L^2})$. Collie and Tong propose RMG${}_{\lambda}$ flow
on $M_n$ with $\lambda=2\pi\kappa$. 
Before examining this flow in detail, we consider Kim and Lee's earlier proposal.

\subsection{Kim-Lee flow on $M_2$}

Motivated by a direct perturbative calculation, Kim and Lee \cite{kimlee}
propose that low-energy vortex dynamics on the euclidean plane, for small $\kappa$, is
described by motion on $M_n$ governed by the Lagrangian
\beq
L=\frac12 \gamma_{L^2}(\dot\alpha,\dot\alpha)+\aa_1(\dot\alpha)+\aa_2(\dot\alpha)
\eeq
where $\aa_1,\aa_2$ are two one-forms on $M_n$, proportional to $\kappa$. This, then, is magnetic geodesic motion on $M_n$
in the effective magnetic field ${\cal B}=\d\aa_1+\d\aa_2$. On $M_n\less\Delta_n$, the 
one forms $\aa_1,\aa_2$ are, in terms of the (unknown)
functions $b_r$,
\bea
\aa_1&=&i\frac{\pi\kappa}{2}\left\{\sum_r(b_rdz_r-\ol{b_r}d\ol{z}_r)-2\sum_{r,s\neq r}\left(
\frac{dz_r}{z_r-z_s}-\frac{d\ol{z}_r}{\ol{z}_r-\ol{z}_s}\right)\right\}\label{keanwa1}\\
\aa_2&=&i\frac{\pi\kappa}{8}\sum_r(H_rdz_r-\ol{H_r}d\ol{z}_r)\label{keanwa2}\\
\mbox{where}\qquad H_r&=&-b_r+\sum_{s\neq r}\left\{(z_r-z_s)\frac{\cd b_r}{\cd z_s}+(\ol{z}_r-\ol{z}_s)\frac{\cd b_r}{\cd \ol{z}_s}
\right\}.
\eea
These formulae simplify considerably in the case $n=2$ (two-vortex dynamics). On $M_2\less\Delta_2$ we define
the centre of mass and relative coordinates 
\beq
Z=\frac12(z_1+z_2),\qquad\qquad\zeta=\sigma e^{i\theta}=\frac12(z_1-z_2)/2
\eeq
respectively. It is known
\cite{sam} that $b_1,b_2$ are functions of $\zeta$ only, and that
\beq\label{kaw}
b_1(\zeta)=b(\sigma)e^{-i\theta}=-b_2(\zeta)
\eeq
where $b(\sigma)$ is some smooth real function on $(0,\infty)$ with the asymptotic behaviour
\beq\label{kelannwat}
b(\sigma)=\frac{1}{\sigma}-\frac12\sigma+O(\sigma^2)
\eeq
as $\sigma\ra 0$. Substituting (\ref{kaw}) into (\ref{keanwa1}),(\ref{keanwa2}) one sees that
\bea
\aa_1&=&2\pi\kappa\left[1-\sigma b(\sigma)\right]\d\theta,\nonumber \\
\aa_2&=&\frac\pi2\kappa\sigma^2b'(\sigma)d\theta.
\eea
It follows that the effective magnetic field is
\beq
{\cal B}=\d(\aa_1+\aa_2)=f(\sigma)d\sigma\wedge\sigma d\theta
\eeq
where 
\beq
f(\sigma)=\frac{\pi\kappa}{\sigma}\frac{d\: }{d\sigma}\left(-2\sigma b(\sigma)+\frac12\sigma^2 b'(\sigma)\right).
\eeq
This formula defines the magnetic field, and hence the flow, on $M_2\less\Delta_2$. In order that the flow be well defined
on the whole of $M_2$, the two-form ${\cal B}$ should extend (at least) continuously to $\Delta_2$. We now show that ${\cal B}$
does {\em not} so extend.
 
Note that, by virtue of (\ref{kelannwat}), $f(\sigma)=\frac32\pi\kappa+O(\sigma^2)$ as $\sigma\ra 0$, that is, 
as the point in $M_2$ approaches 
$\Delta_2$.  Recall \cite{sam} that $\zeta$ is not
a globally well-defined coordinate on $M_2$ because $(Z,\zeta)$ and $(Z,-\zeta)$ correspond to exactly the same
point in $M_2$. To extend any geometric object on $M_2$ over the coincidence set $\Delta_2$, we must use the global
complex coordinates $Z$, $w=\zeta^2$. But then
\beq
{\cal B}=f(|w|^{1/2})\frac{i}{8|w|}dw\wedge d\ol{w}=\frac32\pi\kappa\left(\frac{1}{|w|}+O(1)\right)\frac{i}{8}dw\wedge d\ol{w}
\eeq
as $|w|\ra 0$. Hence ${\cal B}$ blows up on $\Delta_2$, which calls into question the self-consistency of Kim and Lee's
pertrurbative calculation \cite{kimlee}. 

Since $\gamma_{L^2}$ extends smoothly over $\Delta_n$ to give a global k\"ahler metric on $M_n$,  RMG flow is
globally well-defined on $M_n$. It follows that the Kim-Lee flow cannot, as claimed in \cite{colton}, coincide with
RMG flow.

\subsection{The metric on $M_2(\H^2)$}
\label{subsec:H2}

If we wish to study RMG motion of two-vortices on the euclidean plane, we need the coefficient function $b(\sigma)$
introduced in (\ref{kaw}), for which no explicit formula is known (although a conjectural large $\sigma$ asymptotic 
formula is known \cite{manspe}). One must resort to numerics even to construct the metric on $M_2$, therefore \cite{sam}.
Matters improve considerably if we consider vortices moving instead on the {\em hyperbolic} plane with scalar 
curvature $-1$, since the Bogmol'nyi equations (for $\kappa=0$) are then integrable \cite{wit}, and the semi-explicit
formula for $\gamma_{L^2}$ (\ref{hy12}) can, in some nontrivial cases, be made fully explicit \cite{str}. 
In this section we will derive an explicit formula
for the metric on $M_2$. Rather than appealing to (\ref{hy12}) directly, we will analyze the class of k\"ahler metrics on
$M_2$ with the same isometries as $\gamma_{L^2}$. This space of metrics is infinite dimensional, but the $L^2$ metric is
uniquely determined by its restriction to a certain pair of two-dimensional submanifolds of $M_2$, and these restrictions are 
already known \cite{str}. 

Let $\mmm$ denote the double cover of $M_2\less\Delta_2$, that is, $\mmm=\H^2\times\H^2\less\{(z,z)\: z\in\H^2\}$, and
$\wt\gamma$ be the pullback of $\gamma_{L^2}$ to $\mmm$ by the covering map.
It is convenient to use the upper half plane model for $\mathbb{H}^2$, that is, $\H^2=\{x+iy\in\C\::\: y>0\}$. 
with  the Riemannian metric
\begin{equation}\label{sapyou}
g=\frac{2}{y^2}\;(dx^2+dy^2).
\end{equation}
Then there is an isometric action of the projective real linear group  $PL(2,\mathbb{R})$ on $\H^2$, defined by
\begin{equation}\label{hyp0}
z\rightarrow \frac{a z +b}{c z +d}
=\begin{pmatrix} 
a &b\\
c & d
\end{pmatrix}
\odot  z =: M \odot z,
\end{equation}
where $[M]\in PL(2,\mathbb{R})$. This induces an isometric  action  on  $(\mmm,\wt{\gamma})$, 
\begin{equation}\label{hyp1}
(z_1,z_2)\rightarrow (M\odot z_1, M\odot z_2).
\end{equation}
For a generic element of $\mmm$, the isotropy group of $PL(2,\mathbb{R})$ is trivial. By the Orbit-Stabilizer Theorem \cite{Armstrong}, it follows that each generic orbit is diffeomorphic to $PL(2,\mathbb{R})$ itself. Hence, the isometric action of $PL(2,\mathbb{R})$ on $\mmm$ has cohomogeneity $1$, that is, all generic orbits are submanifolds of $\mmm$ with real codimension $1$. Let $s$ denote the distance between two vortices in  $\mathbb{H}^2$. Then each orbit has a unique element $w_s=(ie^{s/2},i e^{-s/2})\in \mmm$. Thus, this action decomposes $\mmm$ into a one parameter family of orbits parameterized by $s>0$, that is, 
$\mmm\equiv (0,\infty) \times PL(2,\R)$.

Consider the coframe  $\{ds,\sigma _k:k=1,2,3\}$ on  $\mmm$ where $\sigma _k$ are the left-invariant $1$-forms dual to the basis $\{e_k:k=1,2,3\}$ of $\mathfrak{p}:=T_{[\mathbb{I}_2]}PL(2,\mathbb{R})$ given by
\begin{equation}\label{hyp2}
e_1=\begin{pmatrix}
0\;\; &1\\
1\;\;& 0
\end{pmatrix},\quad
e_2=\begin{pmatrix}
0 & 1\\
-1 & 0
\end{pmatrix},\quad
e_3=\begin{pmatrix}
1 & 0\\
0 & -1
\end{pmatrix}.
\end{equation}
Any $PL(2,\mathbb{R})$-invariant metric $\tilde{\gamma }$ on  $\mmm$ is determined by a one-parameter family of symmetric bilinear forms $\tilde{\gamma} _s: V_s \times V_s\rightarrow \mathbb{R}$ where $V_s:=\partial /\partial s \oplus \mathfrak{p}$ is the tangent space to $\mmm$ at the element $w_s$.

In terms of the complex  coordinate system $(z_1,z_2)$ on $\mmm$, where $z_1=x_1+iy_1$ and $z_2=x_2+iy_2$, one  can obtain that
\begin{align}
e_1&=(1+e^s)\frac{\partial }{\partial x_1}+ (1+e^{-s})\frac{\partial }{\partial x_2},\quad &
e_2&=(1-e^s)\frac{\partial }{\partial x_1}+ (1-e^{-s})\frac{\partial }{\partial x_2},\nonumber\\
e_3&=2\; e^{s/2}\frac{\partial }{\partial y_1}+2 \; e^{-s/2}\frac{\partial }{\partial y_2},\quad &
\frac{\partial }{\partial s}&=\frac{1}{2}\; e^{s/2} \frac{\partial }{\partial y_1}-\frac{1}{2}\; e^{-s/2} \frac{\partial }{\partial y_2}.\label{hyp3}
\end{align}
Hence, the almost complex structure $J$ on $\mmm$ acts as
\begin{align}\label{hyp6}
Je_1&=\cosh(s/2) \;e_3,\qquad & Je_2&=-4 \sinh(s/2) \;\frac{\partial }{\partial s},\nonumber \\
Je_3&=-\frac{1}{\cosh(s/2)} \;e_1,\qquad & J\frac{\partial }{\partial s}&=\frac{1}{4 \sinh(s/2)}\; e_2.
\end{align}

In addition to the $PL(2,\mathbb{R})$ isometric action on $\mmm$, there is  a discrete isometry on  $\mmm$ defined as $\text{P}:(z_1,z_2)\rightarrow (z_2,z_1)$. Hence, the  group $G:=PL(2,\mathbb{R})\times \{\text{Id},\text{ P}\}$, where $\text{Id}$ is the identity map, acts isometrically on  $\mmm$.

\begin{proposition}\label{hyppro1}\label{mssyc}
Let $\tilde{\gamma} $ be a $G$-invariant k\"{a}hler metric on $\mmm$. Then, there exists a smooth function $A:(0,\infty )\rightarrow \mathbb{R}$ such that
\begin{equation}\label{hyp8}
\tilde{\gamma} =A_1(s) \;ds^2 +A_2(s)\; \sigma _1^2 + A_3(s) \;\sigma _2^2 +A_4(s) \;\sigma _3^2,
\end{equation}
where $A_1(s),\dots,A_4(s)$ are related to $A(s)$ by
\begin{align}
A_1&=\frac{1}{8 \sinh(s/2)} \frac{d}{ds} \biggl(\frac{A(s)}{\cosh(s/2)}\biggr),\quad &A_2&=A(s),\nonumber\\
A_3&=2 \sinh(s/2)\frac{d}{ds}\biggl(\frac{A(s)}{\cosh(s/2)}\biggr),\quad &A_4&=\frac{A(s)}{\cosh^2(s/2)}.\label{hyp9}
\end{align}
\end{proposition}

\begin{proof}
With respect to the coframe $\{ds,\sigma _k:k=1,2,3\}$ on $\mmm$, any $PL(2,\mathbb{R})$-invariant symmetric $(0,2)$ tensor has 
the form
\begin{align}
\tilde{\gamma} &=A_1\; ds^2+ A_2 \;\sigma _1^2 +A_3 \;\sigma_2^2 + A_4\; \sigma _3^2+ 2A_5 \;ds \;\sigma _1 +2 A_6\; ds\; \sigma _2 +2 A_7 \;ds\; \sigma _3\nonumber\\
&\quad + 2 A_8\; \sigma _1 \sigma _2 +2 A_9\; \sigma _1 \sigma _3 + 2 A_{10}\; \sigma _2 \sigma _3,\label{hyp10}
\end{align}
where $A_1,\dots,A_{10}$ are smooth functions of  $s$  only. The metric is also invariant under the isometry $\text{P}$, which 
swaps the two points in $\H^2$. For a given non-coincident pair of points in $\H^2$, this transposition can be accomplished
by acting with an isometry in $PL(2,\R)$. For the point $w_s$, we must act with
\beq
Q=\left(\begin{array}{cc} 0& 1\\ -1 & 0 \end{array}\right),
\eeq
and hence, to swap the pair of points $g\cdot w_s$, where $g\in PL(2,\R)$, we must act with $gQg^{-1}$. Hence, in terms of the
coordinates $s\in (0,\infty)$, $g\in PL(2,\R)$ the isometry $\text{P}$ is
\beq
\text{P}:(s,g)\mapsto (s,gQg^{-1}g)=(s,gQ),
\eeq
that is, right multiplication on $PL(2,\R)$ by $Q$. The induced action of $\text{P}$ on $\p$, identified with
the space of left-invariant vector fields on $PL(2,\R)$, is $X\mapsto \d\text{P}(X)=Q^{-1}XQ$, so $e_1\mapsto -e_1$,
$e_2\mapsto e_2$ and $e_3\mapsto -e_3$. Hence
\beq
\text{P}^*\sigma_1=-\sigma_1, \qquad
\text{P}^*\sigma_2=\sigma_2, \qquad
\text{P}^*\sigma_3=-\sigma_3, \qquad\mbox{and}\qquad
\text{P}^*ds=ds.
\eeq
Now $\text{P}^*\wt\gamma=\wt\gamma$, so we conclude that $A_5=A_7=A_8=A_{10}=0$.

Now, since $\tilde{\gamma }$ is Hermitian, $\tilde{\gamma} _s(u,v)=\tilde{\gamma} _s(Ju,Jv)$ for all $u,v\in V_s$, whence,
using (\ref{hyp6}),
\begin{align}
A_3&\equiv  16 \sinh^2(s/2) A_1,\quad & A_4&\equiv  \frac{A_2}{\cosh^2 (s/2)},\quad & A_9 &\equiv A_6\equiv 0.\label{hyp11}
\end{align}
The k\"ahler form $\omega(\cdot,\cdot)=\tilde{\gamma}(J\cdot,\cdot)$  on $\mmm$ is, therefore,
\beq
\omega=4 \sinh(s/2) A_1 \;ds \wedge  \sigma _2 + \frac{A_2}{\cosh(s/2)}\; \sigma _1\wedge \sigma _3.\label{hyp12}
\eeq
Hence,
\beq
d\omega = -\biggl[ 8 \sinh(s/2) A_1- \frac{d}{ds} \biggl(\frac{A_2}{\cosh(s/2)}\biggr)\biggr]\; ds \wedge \sigma _2 \wedge \sigma _3.\label{hyp20}
\eeq
where we have used the fact that, for our chosen frame/coframe for $PL(2,\R)$,
\beq
\d\sigma_i(e_j,e_k)=e_j[\sigma_i(e_k)]-e_k[\sigma_i(e_j)]-\sigma_i([e_j,e_k])=0-0-\sigma_i([e_j,e_k]|_{\p}).
\eeq
whence
\begin{equation}\label{hyp18}
d\sigma _1= 2 \;\sigma _2 \wedge \sigma _3,\qquad d\sigma _2=2 \;\sigma _1 \wedge  \sigma _3,\qquad d\sigma _3= 2\; \sigma _1\wedge  \sigma _2.
\end{equation}
Since $\tilde{\gamma} $ is k\"{a}hler, $d\omega =0$, so
\beq
\frac{d}{ds} \biggl(\frac{A_2}{\cosh(s/2)}\biggr)-8 \sinh(s/2) A_1 =0.\label{hyp23}
\eeq
So $A_1$, $A_3$, $A_4$ are uniquely determined by the single function $A(\lambda)=A_2(\lambda)$, by (\ref{hyp23}),
(\ref{hyp11}) as claimed.
\end{proof}


\begin{remark} We can equally well think of (\ref{hyp8}) as a formula for a general $PL(2,\R)$ invariant k\"ahler metric
on $M_2$. This space decomposes into a one-parameter family of orbits, parametrized by $s\in[0,\infty)$. Generic orbits
are diffeomorphic to $PL(2,\R)/K$ where $K$ denotes the subgroup $\{\I_2,Q\}$, and there is an exceptional orbit at $s=0$
diffeomorphic to $\H^2$ (the submanifold of coincident vortices). In this picture, one should interpret $\sigma_1^2,\sigma_2^2,
\sigma_3^2$ as $Ad(K)$-invariant symmetric bilinear forms on $\p=T_{[\I_2]}(PL(2,\R)/K)$ (note that $\sigma_1$ and $\sigma_3$ 
are not $Ad(K)$-invariant, so are not well-defined one-forms on $PL(2,\R)/K$). 
\end{remark}

Now, consider the holomorphic isometry of $(\mmm,\wt\gamma)$ defined by
\begin{equation}\label{isoisoP}
 \tilde{K}:(z_1,z_2)\rightarrow  (-\frac{1}{z_2},-\frac{1}{z_1})
\end{equation}
The fixed point set in $\mmm$ of this isometry  is
\begin{equation}
\mmm^0=\{ (\xi ,-\frac{1}{\xi }): i\neq \xi \in \mathbb{H}^2\}\subset \mmm.
\end{equation}
Clearly, $\mmm^0$  is a non-compact $1$-dimensional complex submanifold of $\mmm$. This is the (double cover of the) 
$2$-vortex relative moduli space.  The induced metric on $\mmm^0$ is
\begin{equation}\label{hyp25}
\wt\gamma^0=A_1(s)\;ds^2+A_3(s)\;\sigma _2^2 =A_1(s)\; (ds^2+16 \sinh^2(s/2) \sigma _2^2).
\end{equation}

\begin{corollary}\label{hypmetcor}
The $L^2$ metric on the moduli space $M_2$ is
\begin{equation}
\tilde{\gamma}_{L^2} =A_1(s) \;ds^2 +A_2(s)\; \sigma _1^2 + A_3(s) \;\sigma _2^2 +A_4(s) \;\sigma _3^2,
\end{equation}
where $A_1,\dots,A_4$ are functions of $s$ only determined as in (\ref{hyp9}) by 
\begin{equation}\label{hyp35}
A_{L^2}(s)= 8 \pi  \Bigl( 1+\cosh^2(s/2) + 2 \sqrt{\cosh^2(s/2)+\sinh^4(s/2)} \;\Bigr).
\end{equation}
\end{corollary}

\begin{proof}
The lifted $L^2$ metric on $\mmm$ is a $G$-invariant k\"{a}hler metric, and so is covered by Proposition \ref{hyppro1}. Hence, we only need to determine the function $A(s)$ of the $L^2$ metric. 

An explicit formula for the induced $L^2$ metric on the relative moduli space $\mmm^0$ has been determined by Strachan in \cite{str}, and
rederived and generalized in \cite{kruspe-vortex}, which uses the same conventions for the abelian Higgs model that we are using. Comparing  the formula 
in \cite{kruspe-vortex} with (\ref{hyp25}), we deduce that
\begin{align}
A_1(s)&= 2\pi  \frac{\tanh^2(s/4)}{(1+\tanh^2(s/4))^2}\Biggl[ 1+ \frac{4 (1+\tanh^4(s/4))}{\sqrt{1+\tanh^8(s/4) +14 \tanh^4(s/4)}}\Biggr],\label{notcrap}\\
&=\frac{\pi }{2} \tanh^2(s/2) \Biggl[ 1+\frac{2(\cosh^2(s/2)+1)/\sinh^2(s/2)}{\sqrt{[\cosh(s/2)/\sinh^2(s/2)]^2+1}}\label{crap}
\Biggr].
\end{align}
where we have used the hyperbolic double-angle formulae. Equation (\ref{hyp23}) now gives a differential equation for $A(\lambda)=
A_2$, whose general solution is
\begin{equation}\label{hyp32}
A(s)= 8 \pi  \Bigl( (\cosh^2(s/2) +1) + 2 \sinh^2(s/2) \sqrt{[\cosh(s/2)/\sinh^2(s/2)]^2+1}\; \Bigr)+c \cosh(s/2),
\end{equation}
where $c$ is an integration constant. 

Strachan also gave an explicit formula for the induced $L^2$ metric on $M_n^{co}$, the 
two-dimensional submanifold of $M_n$ consisting of entirely coincident vortices. By symmetry, this metric must be homothetic
to the (physical) metric $g$ on $\H^2$ (\ref{sapyou}). In fact \cite{str}
\beq\label{sysmc}
\gamma_{co}=\frac12\pi n(n+2) g
\eeq
in our conventions.
Consider $X$, the killing vector field on $M_2$ generated by $e_1\in\p$. Its squared length, with respect to $\gamma_{L^2}$,
at the point $(i,i)$ (that is, the coincident two-vortex positioned at $i\in\H^2$) is, by definition, $A_2(0)$. $X((i,i))$ 
is clearly tangent to $M_2^{co}$ and moves the coincident two-vortex through
$x+iy=i$ with initial velocity vector $2\cd/\cd y$, and hence with
squared speed $8$ (with respect to the metric $g$).
Hence, by (\ref{sysmc}), 
\beq
\gamma_{co}(X,X)=32\pi.
\eeq
Comparing with (\ref{hyp32}), we
deduce that $c=0$, which completes the proof.
\end{proof}


\begin{proposition}\label{hyppro2}
Let $\tilde{\gamma} $ be a $G$-invariant K\"{a}hler metric on $\mmm$, determined as in Proposition  \ref{hyppro1} by
a function $A(s)$. Then, its Ricci curvature  tensor is
\begin{equation}\label{hyp37}
\text{Ric} =C_1(s) \;ds^2 +C_2(s) \;\sigma _1^2 + C_3(s) \;\sigma _2^2 +C_4(s)\; \sigma _3^2,
\end{equation}
where $C_1,\dots,C_4$ are smooth functions of  $s$ only, defined as in (\ref{hyp9}), by a single function $C(s)$ given by
\begin{equation}\label{hyp38}
C(s)=-4 \sinh(s/2) \cosh(s/2) \frac{d}{ds} \log \biggr( \frac{A_1 A_2}{ \cosh^2(s/2)}\biggl)-8 \cosh^2(s/2).
\end{equation}
\end{proposition}

\begin{proof}
The Ricci curvature  tensor with respect to $\tilde{\gamma}$ is a $G$-invariant symmetric $(0,2)$ tensor on $\mmm$ which is hermitian and whose associated $2$-form $\rho (\cdot,\cdot)=\text{Ric}(J\cdot,\cdot)$, the Ricci form, is closed. Thus, it is covered by 
Proposition  \ref{hyppro1}, that is,  $\text{Ric}$  has the same structure as $\tilde{\gamma} $ and its coefficients $C_1,\dots,C_4$  are related, as in (\ref{hyp9}), to the function $C(s):=C_2(s)=\text{Ric}_s(e_1,e_1)$. Introducing an  orthonormal basis $\{E_k:k=0,1,2,3\}$ of  $(V_s,\gamma _s)$ as
\begin{equation}\label{hyp39}
E_0=\frac{1}{\sqrt{A_1}} \frac{\partial }{\partial s},\quad E_1=\frac{1}{\sqrt{A_2}} \;e_1,\quad E_2=\frac{1}{\sqrt{A_3}} \;e_2,\quad E_3=\frac{1}{\sqrt{A_4}}\; e_3,
\end{equation}
then, by  the definition of the Ricci curvature tensor, we obtain that 
\begin{align}
\text{Ric}_s(e_1,e_1)&=\sum_{i=0}^{3} \tilde{\gamma}_s (R(E_i,e_1) e_1, E_i),\nonumber\\
&=-4 \sinh(s/2) \cosh(s/2) \frac{d}{ds} \log \biggr( \frac{A_1 A_2}{ \cosh^2(s/2)}\biggl)-8 \cosh^2(s/2),\label{hyp40}
\end{align}
where $R$ is the Riemannian curvature tensor associated with $\tilde{\gamma} $. Hence, the claim is proved.
\end{proof}

The Ricci form $\rho $ on $M_2$  has the same structure as the k\"{a}hler form $\omega$,  that is,
\beq\label{hyp41}
\rho =4 \sinh(s/2) C_1 \; ds \wedge \sigma _2 + \frac{C_2}{\cosh(s/2)}\; \sigma _1\wedge \sigma _3.
\eeq
Since $M_2$ has trivial second cohomology, this (closed) form must be exact.
Rewriting $C_1$ in terms of $C$ one sees that
\beq
\rho=\d\left(\frac{C(s)}{2\cosh(s/2)}\:\sigma_2\right).
\eeq


\subsection{RMG motion on the reduced moduli space}

The $2$-vortex relative moduli space  $M_2^0$ is the fixed point set in  $\mmm$ of the holomorphic isometry $\tilde{K}$, defined in (\ref{isoisoP}). Hence, by Corollary \ref{totrmg3},   RMG curves with initial data in $TM_2^0$ remain  on $M_2^0$ for all time. So,  RMG flow localizes to $M_2^0$. However, the restriction of  the Ricci form $\rho $  on  $M_2$ to $M_2^0$ does not coincide with $\rho^0 $, the Ricci form on $M_2^0$ defined by its induced metric $\gamma^0$. Hence, the RMG flow on $M_2^0$, thought of as a submanifold of $M_2$,  does not coincide with the  RMG flow on $M_2^0$, thought of as a k\"{a}hler manifold in its own right. Here, we will compare these flows on $M_2^0$, which we call the extrinsic and intrinsic RMG flow, respectively.

It follows from Corollary \ref{hypmetcor} and Proposition \ref{hyppro2} that the restricted and
intrinsic Ricci forms on $M_2^0$ are
\beq
\rho\rvert=F\vert(s)ds\wedge\sigma_2,\qquad
\rho^0=F^0(s)ds\wedge\sigma_2
\eeq
where
\bea
F\rvert(s)&=& F^0(s)-\frac{d}{ds} \Bigl[2 \sinh(s/2) \frac{d}{ds} \log \Bigl(\frac{A(s)}{\cosh^2(s/2)}\Bigr)\Bigr]-\sinh(s/2)\\
F^0(s)&=& -\frac{d}{ds} \Bigl[ 2 \sinh(s/2) \frac{d}{ds} \log \left(\frac{d}{ds}\Bigl(\frac{A(s)}{\cosh(s/2)}\Bigr)\right)\Bigr]-\frac{1}{2} \sinh(s/2)
\eea
In the case of the $L^2$ metric, these functions behave asymptotically like
\begin{align}
F\rvert&\sim -\frac{1}{5}s^3,\qquad {F^0 }\sim  \frac{7}{40}s^3,\quad\;\;\;\text{as}\; s\rightarrow 0,\label{indric22}\\
F\rvert&\sim -e^{s/2},\qquad {F^0 }\sim  -\frac{1}{2}e^{s/2},\quad\text{as}\; s\rightarrow \infty .\label{indric2}
\end{align}
From (\ref{indric2}), one can see that  even as $s\rightarrow \infty $,  the restricted and intrinsic Ricci forms do not coincide. Comparing $F\rvert$ with ${F^0 }$ in (\ref{indric2}), one expects that the extrinsic and intrinsic RMG flows  coincide  for large $s$ if the RMG parameters  in each are  related by
\begin{equation}\label{exin}
\lambda _{\text{extrinsic}}=\frac{1}{2} \lambda _{\text{intrinsic}}.
\end{equation}
Henceforth, when comparing the two flows we choose their parameters to be related in this fashion.
In this case, one expects the RMG flows in the core region of $M_2^0$ (i.e.\ for $s$ small) to be
qualitatively quite different, since $F|$ is uniformly negative, while $F^0$ is positive for $s$ 
small and negative for $s$ large. Figure \ref{FFcomp} compares $F\vert$ and $2F^0$.

\begin{figure}[H]
\centering
\includegraphics[trim=1.5cm 3.5cm 1cm 2cm, clip,width=0.4\textwidth,height=0.3\textheight ]{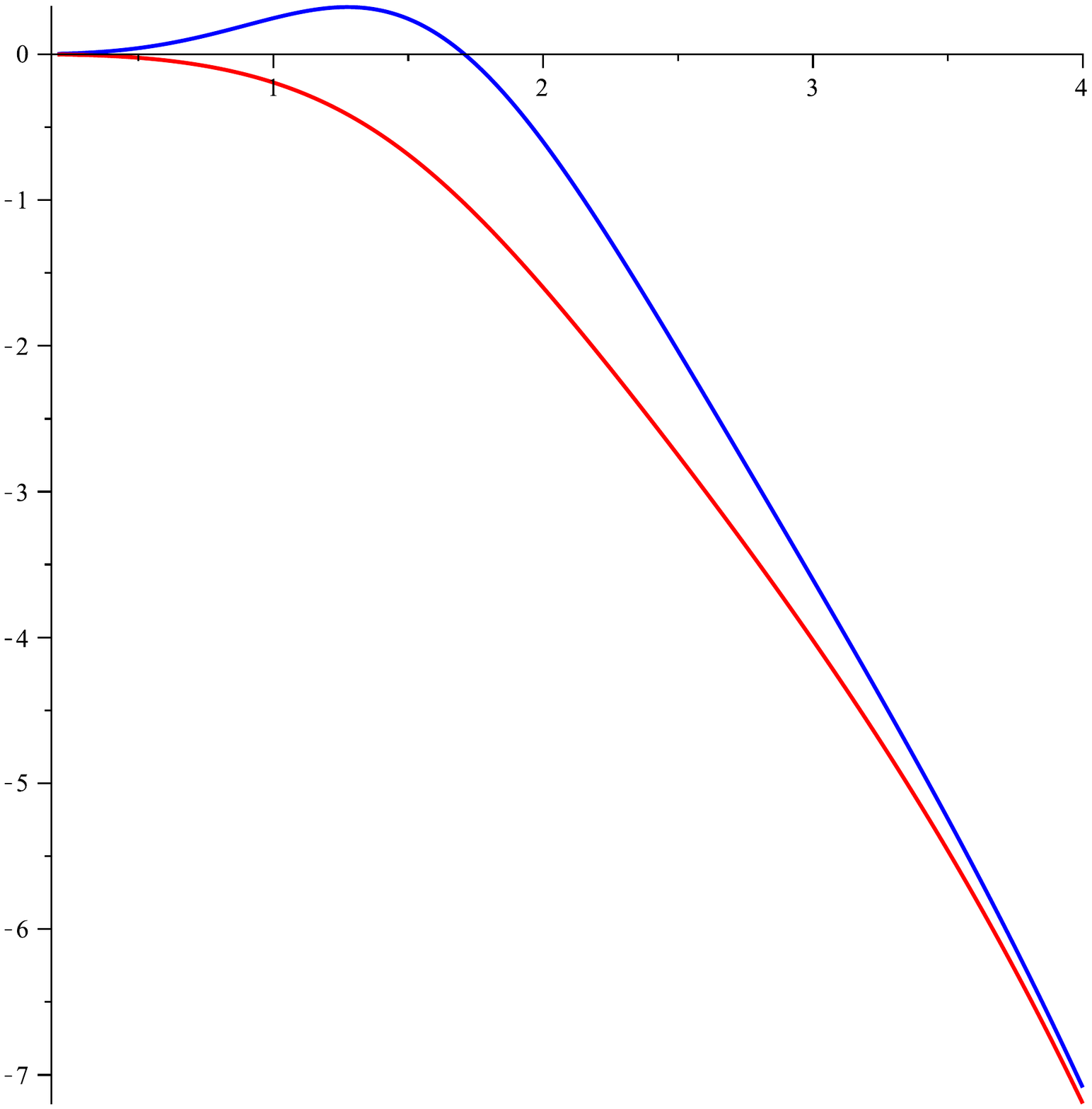}
\caption{\small \sf Comparison of the restricted and intrinsic Ricci forms on the space
of centred hyperbolic two vortices: plots of $F\vert(s)$ (red) and $2F^0(s)$ (blue).}
\label{FFcomp}
\end{figure}

Magnetic geodesic flow on $M_2^0$ with respect to a rotationally invariant effective magnetic field
${\cal B}=F(s)ds\wedge\sigma_2$ is governed by the ODE system
\begin{align}
\ddot{s}&=-\frac{1}{ 2 A_1(s)} \bigl [ A_1'(s) \dot{s}^2 -A_3'(s) \dot{\psi  }^2 + 2F(s)\dot{\psi  }\bigr],\nonumber\\
\ddot{\psi }&= -\frac{1}{A_3(s)} \bigl [A_3'(s) \dot{s}\;\dot{\psi  } -F(s) \dot{s}\bigr],\label{RMGequations}
\end{align}
where $\psi$ is an angular coordinate chosen so that $d\psi  =\sigma _2\rvert$. Choosing
$F(s)$ to be $F\rvert(s)$ or $2F^0(s)$ we obtain the extrinsic and intrinsic RMG flows,
respectively, normalized so as to coincide asymptotically at large $s$.  We have solved
these ODE systems numerically 
for  various initial  values. The corresponding RMG trajectories of one of the vortices  on the Poincar\'{e } disk are depicted in Figure \ref{RMGflowhypM22}.  As expected, RMG trajectories which
reach the core region of $M_2^0$ exhibit marked differences in the two flows. 

\begin{figure}[H]
\centering
\begin{tabular}{cccc}
\includegraphics[trim=4cm 4cm .5cm 5cm, clip, height=0.2\textheight]{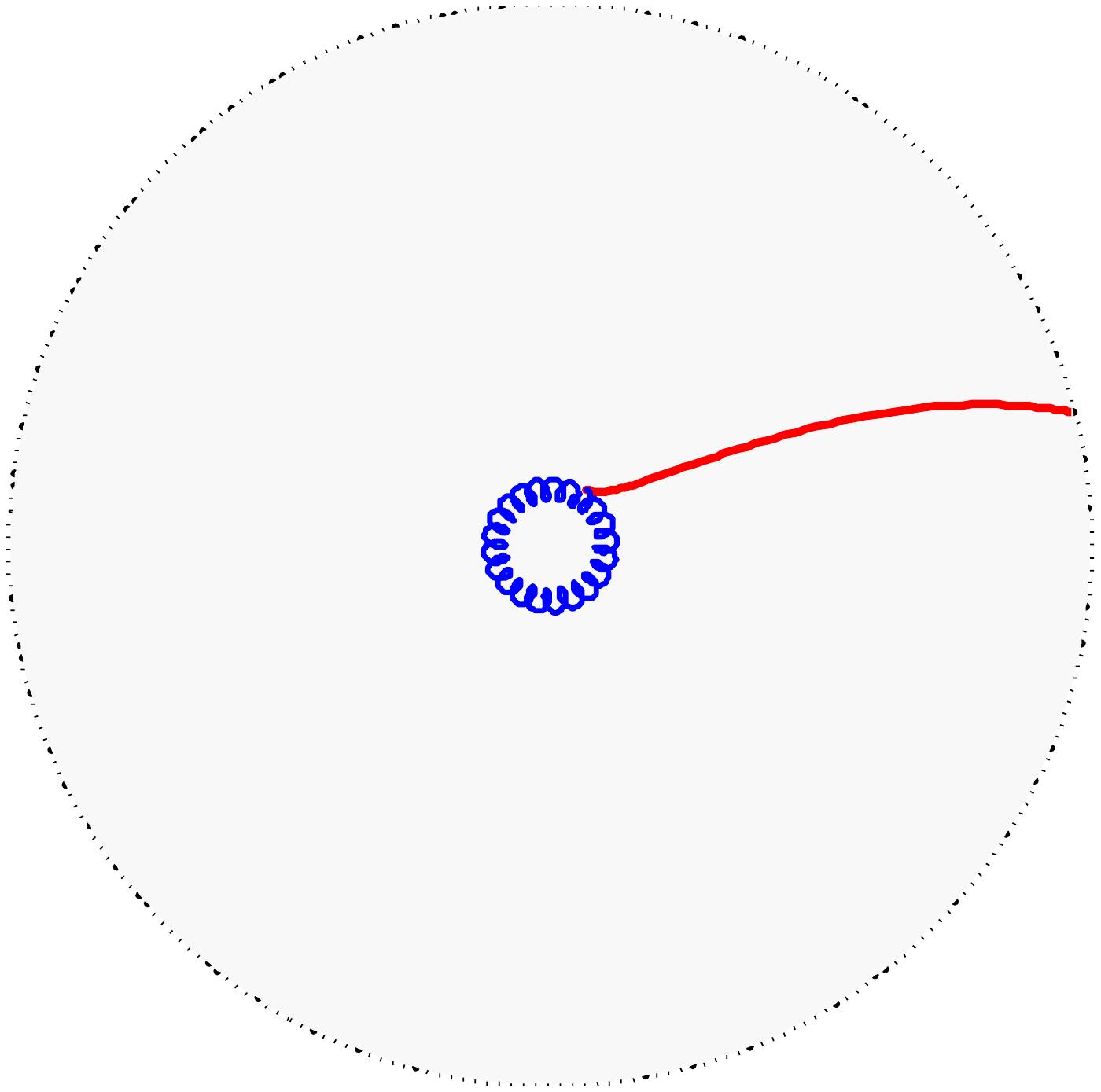}  &
\includegraphics[trim=4cm 4cm .5cm 5cm, clip, height=0.2\textheight]{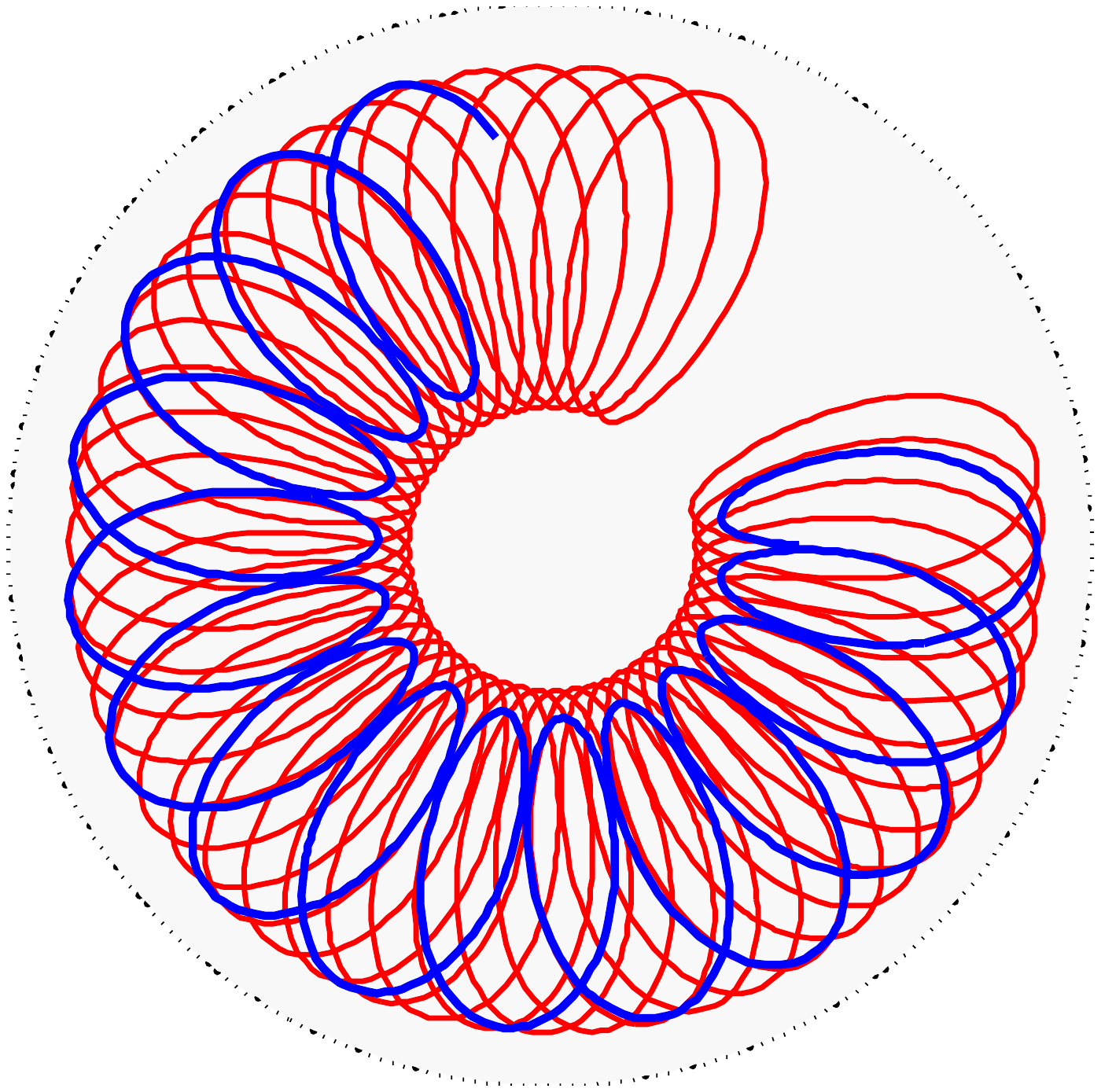}  &
\includegraphics[trim=4cm 4cm .5cm 5cm, clip, height=0.2\textheight]{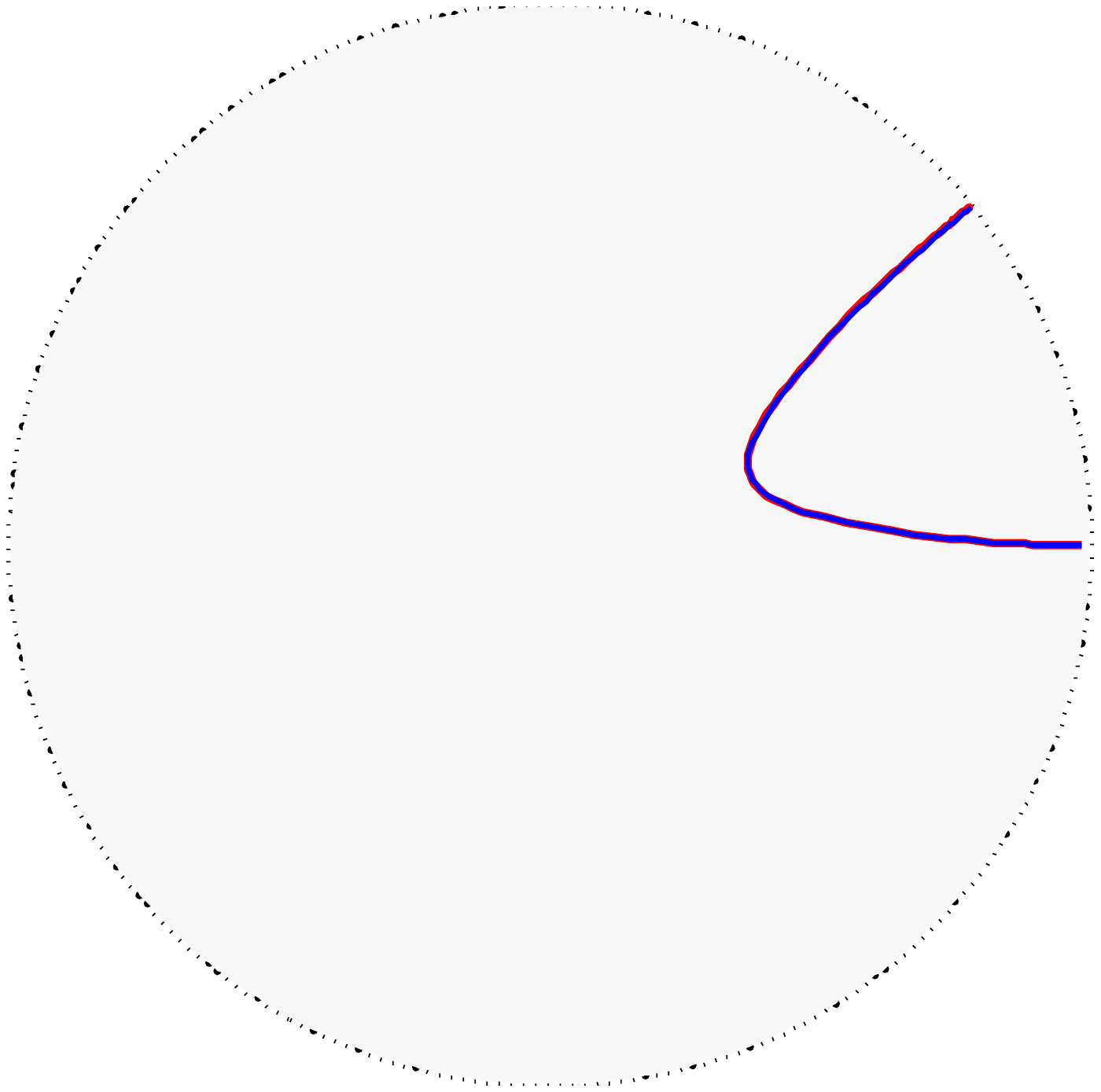}  \\
({\sf a}) & ({\sf b}) & ({\sf c})\\
\end{tabular}
\caption{\small \sf Plots of  vortex  trajectories under the extrinsic RMG flow on $M_2^0$ (red) and the intrinsic RMG flow  (blue) with  $\lambda_{\text{intrinsic}}= 2 \lambda _{\text{extrinsic}}$ and various initial values .}
\label{RMGflowhypM22}\label{lnagf}
\end{figure}

A related observation concerns orbiting vortex pairs. It is immediate from (\ref{RMGequations}) that
a circle of constant $s$ is a closed magnetic geodesic if and only if it is traversed at constant
angular velocity
\beq
\dot\psi=\nu(s)=2\frac{F(s)}{A_3'(s)}.
\eeq
This, then, gives the frequency-separation relation for a pair of vortices orbiting one another
at fixed separation. Note that $A_3'(s)>0$ for all $s>0$.
From figure \ref{FFcomp}, one sees that (for $\lambda>0$) vortex pairs obeying
the intrinsic RMG flow orbit
one another anticlockwise for $s<s_0\approx 1.7$ 
and clockwise for $s>s_0$, whereas orbiting vortex pairs always circulate clockwise in the
extrinsic flow.

Unlike geodesics, the features of  RMG flow on a fixed point set of  a holomorphic isometry,   such as $M_2^0$,  cannot be deduced by knowing only the metric on the fixed point set. This  difference makes  RMG flow  significantly harder to study than geodesic flow and means that studies of intrinsic
RMG flow on low-dimensional submanifolds, such as those presented in \cite{kruspe-vortex}, are of limited
value in trying to understand the true (extrinsic) RMG flow.


\section{RMG motion of $\mathbb{C}P^1$ lumps}
\label{sec:lump}
\news

As observed in section \ref{sec:intro}, the question of {\em completeness} of RMG flow on a noncompact k\"ahler manifold is
interesting and nontrivial. Certainly, if the manifold is complete (as a metric space or, equivalently, its geodesic flow is
complete), then it is RMG complete since RMG flow conserves speed (so an RMG curve which escaped every compact set in bounded
time would define a divergent Cauchy sequence). Since RMG flow converges (pointwise) to geodesic flow in the limit of large speed,
it has been conjectured that the converse holds also: if a k\"ahler manifold is RMG complete then it is geodesically complete
\cite{kruspe-vortex}. In fact this is false, and in this section we provide a counterexample of independent interest: the moduli 
space of unit charge $\CP^1$ lumps on $S^2$, or, equivalently, the space $\rat_1$ of degree one rational maps $S^2\ra S^2$, 
given its $L^2$ metric.

Recall that $\rat_n$ is the space of degree $n$ holomorphic maps $S^2\ra S^2$. If one chooses stereographic coordinates
$z,W$ on the domain and codomain, such maps take the form
\beq
W(z)=\frac{a_0+a_1z+\cdots+a_nz^n}{b_0+b_1z+\cdots+b_n z^n}
\eeq
where $a_0,\ldots,a_n,b_0,\ldots, b_n\in\C$ are constant. There is a natural inclusion $\rat_n\hra\CP^{2n+1}$ defined by
$W(z)\mapsto[a_0,\ldots,a_n,b_0,\ldots,b_n]$, which equips $\rat_n$ with a complex structure, and a natural metric on
$\rat_n$ defined by restricting the $L^2$ norm on $W^{-1}TS^2$ to $T\rat_n$.
It is known \cite{spe-l2} that $\rat_n$ is k\"ahler with respect to this metric 
and complex structure. Further, $\rat_1$ is diffeomorphic to $SO(3)\times\R^3$. One may regard ${\cal O}\in SO(3)$
as parametrizing the internal orientation of the lump and $\lamvec\in\R^3$ as parametrizing both its sharpness, 
$\lambda=|\lamvec|$,
and its position in physical space $-\lamvec/|\lamvec|\in S^2$. Lumps with $\lamvec=\zerovec$ have sharpness $0$, that is,
constant energy density, so do not have a well-defined position. Explicitly, 
the point $(\I_3,(0,0,\lambda))$ corresponds to the rational map
\beq
W(z)=\mu(\lambda)z,\qquad\mbox{where}\qquad \mu=\frac{\Lambda+\lambda}{\Lambda-\lambda}\qquad\mbox{and}\qquad
\Lambda=\sqrt{1+\lambda^2},
\eeq
and every other point in $\rat_1$ can be reached from a point such as this by acting with some isometry:
$G=SO(3)\times SO(3)$ acts isometrically on $\rat_1$ via
\beq
({\cal L},{\cal R}):({\cal O},\lamvec)\mapsto ({\cal LOR}^{-1},{\cal R}\lamvec).
\eeq
This is just the restiction to $\rat_1$ of the natural action of $G$ on all smooth maps $S^2\ra S^2$, namely,
$({\cal L},{\cal R}):\phi\mapsto {\cal L}\circ\phi\circ{\cal R}^{-1}$.

$G$-invariance and the k\"ahler property almost completely determine $\gamma_{L^2}$. By an argument similar to that used
to prove Proposition \ref{mssyc}, one finds \cite{spe-l2} that
\begin{equation}\label{metric}
\gamma_{L^2}   =A_1 d\boldsymbol{\lambda }\cdot d\boldsymbol{\lambda }+A_2 (\boldsymbol{\lambda }\cdot d\boldsymbol{\lambda })^2+A_3 \boldsymbol{\sigma  }\cdot \boldsymbol{\sigma  }+A_4(\boldsymbol{\lambda }\cdot \boldsymbol{\sigma  })^2+A_5 \boldsymbol{\lambda }\cdot (\boldsymbol{\sigma }\times d\boldsymbol{\lambda }),
\end{equation}
where  $A_1,\dots,A_5$ are smooth functions of $\lambda$ only defined by the single function
\begin{equation}\label{A}
A(\lambda )=2\pi \mu \;\frac{ [\mu ^4-4\mu ^2 \log\mu -1]}{(\mu ^2-1)^3},
\end{equation}
 as follows
\begin{align}
A_1&=A(\lambda ),\qquad A_2=\frac{A(\lambda )}{\Lambda ^2}+\frac{A'(\lambda )}{\lambda },\qquad A_3=(\frac{1+2\lambda ^2}{4})A(\lambda ),\nonumber\\
A_4&=\biggl(\frac{\Lambda^2 }{4\lambda }\biggr)A'(\lambda ),\qquad A_5=A(\lambda ).\label{function}
\end{align}
In (\ref{metric}) $\sigvec=(\sigma_1,\sigma_2,\sigma_3)$ is the triple of left invariant one forms on $SO(3)$ dual 
to the left invariant vector fields $\theta_1,\theta_2,\theta_3$ which, at the identity $\I_3$,
 coincide with the usual basis for $\so(3)$,
that is
\beq
E_1=\left(\begin{array}{ccc}0&0&0\\0&0&-1\\0&1&0\end{array}\right), \qquad
E_2=\left(\begin{array}{ccc}0&0&1\\0&0&0\\-1&0&0\end{array}\right), \qquad
E_3=\left(\begin{array}{ccc}0&-1&0\\1&0&0\\0&0&0\end{array}\right).
\eeq
Also,  $\times$ denotes vector product in $\R^3$, and 
juxtaposition of one-forms denotes symmetrized tensor product. Note that, although we are using similar
notation to that of section \ref{subsec:H2}, the functions $A_i$ and one-forms $\sigma_i$ in (\ref{metric}) are unrelated to the
analogous quantities defined there.  It follows immediately from (\ref{metric}) that $\rat_1$ is
geodesically incomplete (for example, the curve $(\I_3,(0,0,s))$ with $s\in\R$ has finite length). 

Since $\rat_1$ has trivial second cohomology, its Ricci form $\rho$ is necessarily exact. An explicit formula for $\rho$
was derived in \cite{spe-l2}, from which it quickly follows that
\begin{equation}\label{1form}
\rho=\d\aa,\qquad \aa=\frac{\Lambda }{2}\bar{A}(\lambda )\:\;(\boldsymbol{\lambda }\cdot \boldsymbol{\sigma }),
\end{equation}
where
\begin{equation}\label{barA}
\bar{A}(\lambda )=-\frac{1}{2\lambda }\frac{d}{d\lambda } \log(A^2(\lambda )\;B(\lambda )),
\end{equation}
and 
\begin{equation}\label{B}
B(\lambda ):=A_3(\lambda )+\lambda ^2 A_4(\lambda )=\frac{1+2\lambda ^2}{4} A(\lambda )+\frac{\lambda \Lambda ^2}{4} A'(\lambda )=\frac{\Lambda }{4} \frac{d}{d\lambda } (\lambda  \Lambda  A(\lambda )).
\end{equation}
Hence RMG flow on $\rat_1$ is governed by the Lagrangian
\beq
L=\frac{1}{2}[ A_1 (\dot{\boldsymbol{\lambda }} \cdot \dot{\boldsymbol{\lambda }}) +A_2 (\boldsymbol{\lambda } \cdot \dot{\boldsymbol{\lambda }})^2 +A_3 (\boldsymbol{\Omega  } \cdot \boldsymbol{\Omega }) +A_4 (\boldsymbol{\lambda } \cdot \boldsymbol{\Omega })^2 +A_5 \boldsymbol{\lambda } \cdot (\boldsymbol{\Omega }\times \dot{\boldsymbol{\lambda }})
 -  \Lambda \bar{A} (\boldsymbol{\lambda } \cdot \boldsymbol{\Omega })]\label{L}
\eeq
for a curve $\chi (t)=({\cal O}(t), \boldsymbol{\lambda }(t))$ on $\rat_1$ whose angular velocity $\Omegavec\in\R^3$ is defined
such that
${\cal O}(t)^{-1}\dot{\cal O}(t)=\Omegavec\cdot\Evec\in\so(3)$.
Note that, to avoid confusion with the radial coordinate $\lambda=\|\lamvec\|$, we have used the scaling property of
RMG flow to set the effective electric charge (denoted $\lambda$ in section \ref{sec:intro}) to unity.
This flow conserves total energy
\begin{equation}\label{E}
E=\frac{1}{2}[ A_1 (\dot{\boldsymbol{\lambda }} \cdot \dot{\boldsymbol{\lambda }}) +A_2 (\boldsymbol{\lambda } \cdot \dot{\boldsymbol{\lambda }})^2 +A_3 (\boldsymbol{\Omega  } \cdot \boldsymbol{\Omega }) +A_4 (\boldsymbol{\lambda } \cdot \boldsymbol{\Omega })^2 +A_5 \boldsymbol{\lambda } \cdot (\boldsymbol{\Omega }\times \dot{\boldsymbol{\lambda }})].
\end{equation}
Furthermore, we have

\begin{proposition}\label{iwtssyc}
RMG flow on $\rat_1$ conserves the angular momenta $\{P_k,Q_k:k=1,2,3\}$ given by
\begin{align}
P_k&=\sum_{j=1,2,3}{\cal O}_{jk} [A_3 \Omega_j +A_4(\boldsymbol{\Omega} \cdot  \boldsymbol{\lambda }) \lambda_j +\frac{1}{2} A_1 (\dot{\boldsymbol{\lambda} } \times \boldsymbol{\lambda}))_j- \frac{1 }{2} \Lambda \bar{A}\lambda_j ],\label{P}\\
Q_k&=(A_3-\frac{1}{2} \lambda ^2 A_1) \Omega_k +(A_4+\frac{1}{2} A_1) (\boldsymbol{\Omega} \cdot \boldsymbol{\lambda }) \lambda_k  -\frac{1}{2} A_1 (\dot{\boldsymbol{\lambda }}\times \boldsymbol{\lambda })_k- \frac{1 }{2} \Lambda \bar{A} \lambda_k.\label{Q}
\end{align}
\end{proposition}

\begin{proof}
The RMG Lagrangian, given in (\ref{L}), has $G=SO(3)\times SO(3)$ symmetry.  Hence, there is a set of six conserved angular 
momenta, one for each generator of $G$. Given $\bar{Y}\in\g=\so(3)\oplus\so(3)$, denote by $Y$ the killing vector field on $\rat_1$
which it induces. Then the conserved Noether charge associated with the infinitesimal symmetry $\bar{Y}$ is
\cite{mansut}
\begin{equation}\label{iltwsyjo}
\text{J}_Y=\gamma_{L^2} (Y,\dot{\chi })-\aa(Y)+\alpha_Y,
\end{equation}
where $\alpha_Y$ is a real function  on $\rat_1$ such that $\d\alpha_Y=\mathcal{L}_Y\aa$. Since $\aa$ is $G$-invariant, 
$\mathcal{L}_Y\aa=0$ for all $\bar{Y}$, so we may take $\alpha_Y=0$ for all $\bar{Y}$.

The six killing vector fields on $\rat_1= SO(3)\times\R^3$ generated by the usual basis for $\g$ are \cite{kruspe-lump}
\beq
\xi_k,\qquad Z_k=\theta_k+\sum_{i,j}\epsilon_{kij}\lambda_i\frac{\cd\: }{\cd \lambda_j},\qquad k=1,2,3
\eeq
where $\theta_i$ are, as before, the left invariant vector fields on $SO(3)$ dual to $\sigma_i$, and $\xi_i$ are the
{\em right} invariant vector fields on $SO(3)$ with $\xi_i(\I_3)=\theta_i(\I_3)$, explicitly, 
\begin{equation}\label{basis1}
\xi _k=\sum_j{\cal O}_{jk}\;\theta _j.
\end{equation}
Setting $Y=\xi_k$ in (\ref{iltwsyjo}) yields the conserved charge $\text{J}_Y=P_k$ claimed, and similarly
setting $Y=Z_k$ yields the charge $Q_k$.
\end{proof}

It is convenient to collect the angular momenta  $\{P_k:k=1,2,3\}$ and $\{Q_k:k=1,2,3\}$ into a pair of 3-vectors
\begin{align}
\boldsymbol{P}&={\cal O}^T [A_3 \boldsymbol{\Omega} +A_4(\boldsymbol{\Omega} \cdot  \boldsymbol{\lambda }) \boldsymbol{\lambda} +\frac{1}{2} A_1 (\dot{\boldsymbol{\lambda} } \times \boldsymbol{\lambda}))- \frac{1 }{2} \Lambda \bar{A} \boldsymbol{\lambda} ],\label{P} \\
\boldsymbol{Q}&=(A_3-\frac{1}{2} \lambda ^2 A_1) \boldsymbol{\Omega} +(A_4+\frac{1}{2} A_1) (\boldsymbol{\Omega} \cdot \boldsymbol{\lambda }) \boldsymbol{\lambda}  -\frac{1}{2} A_1 (\dot{\boldsymbol{\lambda }}\times \boldsymbol{\lambda })- \frac{1 }{2} \Lambda \bar{A} \boldsymbol{\lambda}.\label{Q}
\end{align}
Having determined the conserved quantities  $E$,  $\boldsymbol{P}$ and $\boldsymbol{Q}$ associated with the RMG flow on $\rat_1$, 
one can  eliminate $\boldsymbol{\Omega }$ from $E$ to obtain
\begin{align}
E=& \frac{1}{2}\biggl( A_1 \| \dot{\boldsymbol{\lambda }} \| ^2 +A_2 (\boldsymbol{\lambda } \cdot \dot{\boldsymbol{\lambda }})^2 +\frac{1}{A_3}  \|  \boldsymbol{P} \|  ^2 -\frac{A_1^2}{4 A_3}  \|  \dot{\boldsymbol{\lambda }}\times \boldsymbol{\lambda} \|  ^2 + \frac{ \Lambda \bar{A} }{A_3} (\boldsymbol{Q} \cdot \boldsymbol{\lambda })\nonumber \\
&- \frac{A_4}{ A_3 B}\;[(\boldsymbol{Q} \cdot \boldsymbol{\lambda })+\frac{ 1}{2}\lambda ^2 \Lambda\bar{A}]^2+\frac{\lambda ^2 \Lambda^2 \bar{A}^2}{4 A_3} \biggr).\label{EE}
\end{align}

Consider the mapping $q:T\rat_1\ra\R\times\R^3\times\R^3$ which assigns to each tangent vector the triple 
$(E,\boldsymbol{P},\boldsymbol{Q})$. By Proposition \ref{iwtssyc}, every RMG curve in $\rat_1$ is confined to some
level set of $q$. That RMG flow is complete will follow quickly from the following:

\begin{thm}\label{xcomp}
 Every level set of $q:T\rat_1\ra\R\times\R^3\times\R^3$ is compact.
\end{thm}

\begin{proof} Choose and fix $(E,\boldsymbol{P},\boldsymbol{Q})\in\R^7$ and let $X=q^{-1}(E,\boldsymbol{P},\boldsymbol{Q})
\subset T\rat_1$. Now $T\rat_1\equiv TSO(3)\times T\R^3$, and $TSO(3)\equiv SO(3)\times\R^3$ via the identification 
$\dot{\cal O}\mapsto\Omegavec$. We may realize $SO(3)$ as a submanifold of $\R^9$ by mapping ${\cal O}$ to its list of
matrix elements. In this way, we may regard $T\rat_1$ as a 12-dimensional submanifold of $\R^9\times\R^3\times\R^3\times\R^3$.
The mapping $q$ is smooth, hence certainly continuous, so $X$ is closed. Hence, by Heine-Borel, it suffices to show that
$X\subset\R^{18}$ is bounded in euclidean norm.

Assume, towards a contradiction, that there is some sequence $x_n=({\cal O}_n,\Omegavec_n,\lamvec_n,\dot\lamvec_n)\in X$
which is unbounded in euclidean norm. By the definition of $SO(3)$, $\|{\cal O}_n\|_{\R^9}=\sqrt{3}$ for all $n$, so 
(at least) one of
$\lamvec_n,\dot\lamvec_n,\Omegavec_n$ must be unbounded. We will now eliminate these possibilities in turn.

Assume $\lamvec_n$ is unbounded. Then it has a subsequence, which we still denote $\lamvec_n$, with
$\|\lamvec_n\|\ra\infty$.
For all $\lamvec\neq\zerovec$ let $\hat\lamvec=\lamvec/\|\lamvec\|$ and
\begin{align}
H(\boldsymbol{\lambda },\dot{\boldsymbol{\lambda }})&:=A_1 \| \dot{\boldsymbol{\lambda }}\| ^2 + \lambda ^2 A_2 (\hat{\boldsymbol{\lambda }} \cdot \dot{\boldsymbol{\lambda }})^2 -\frac{ \lambda ^2 A_1^2}{4 A_3} \| \hat{\boldsymbol{\lambda }}\times  \dot{\boldsymbol{\lambda }}\| ^2,\label{H}\\
G(\boldsymbol{\lambda },\boldsymbol{Q})&:=F_1(\lambda )\; (\boldsymbol{Q} \cdot \hat{ \boldsymbol{\lambda }})^2 \;+ F_2(\lambda )\;(\boldsymbol{Q} \cdot \hat{\boldsymbol{\lambda }}) \;+F_3(\lambda ),\label{G}
\end{align}
where 
\begin{equation}\label{FFF}
F_1(\lambda )=-\frac{\lambda ^2 A_4}{A_3 B},\qquad  F_2(\lambda )=\frac{\lambda  \Lambda \bar{A}}{A_3}\biggl(1-\frac{\lambda ^2 A_4}{B}\biggr), \qquad F_3(\lambda )=\frac{\lambda ^2 \Lambda ^2 \bar{A}^2}{4 B}.
\end{equation}
Then, the conserved energy $E$, given in (\ref{EE}), can be written as
\begin{equation}\label{EEEE}
 E = \frac{1}{2}\biggl(H(\boldsymbol{\lambda },\dot{\boldsymbol{\lambda }}) +\frac{1}{A_3}\|  \boldsymbol{P} \|  ^2 + G(\boldsymbol{\lambda },\boldsymbol{Q})\biggr).
\end{equation}
Since the cross and dot products on $\mathbb{R}^3$ are related by
\begin{equation}\label{L-L}
\| \hat{\boldsymbol{\lambda }} \times  \dot{\boldsymbol{\lambda }}\| ^2= \|\dot{\boldsymbol{\lambda }}\| ^2-(\hat{\boldsymbol{\lambda }} \cdot \dot{\boldsymbol{\lambda }})^2,
\end{equation} 
then,
\begin{align}
H(\boldsymbol{\lambda} ,\dot{\boldsymbol{\lambda }})&=\frac{\Lambda ^2 A}{1+2\lambda ^2}\;   \|\dot{\boldsymbol{\lambda }}\|^2 + \lambda ^2\biggl(\frac{2+3\lambda ^2}{(1+2\lambda ^2) \Lambda ^2} A+\frac{A'(\lambda )}{\lambda }\biggr ) (\hat{\boldsymbol{\lambda }} \cdot \dot{\boldsymbol{\lambda }})^2.\label{H-H}
\end{align}
Here,   we have used the definition of $A_1,A_2$ and $A_3$, as in (\ref{function}). Since $B(\lambda )$, given in (\ref{B}), is positive,
\begin{equation}\label{A-A}
\frac{A'(\lambda )}{\lambda } > - \frac{1+2\lambda ^2}{\lambda ^2 \Lambda ^2} A,
\end{equation}
from which it follows that 
\begin{equation}\label{H}
H(\boldsymbol{\lambda },\dot{\boldsymbol{\lambda }})\geq  \frac{\Lambda ^2 A}{1+2\lambda ^2}\; \| \hat{\boldsymbol{\lambda }} \times \dot{\boldsymbol{\lambda} } \| ^2\geq 0.
\end{equation}
Since both $H(\boldsymbol{\lambda },\dot{\boldsymbol{\lambda }})$  and $A_3 $ are non-negative, it follows from (\ref{EEEE}) that
\begin{equation}\label{inqEG}
2 E \geq   G(\boldsymbol{\lambda }_n,\boldsymbol{Q}).
\end{equation}
From (\ref{A}), one obtains the following limit 
\begin{equation}\label{limG}
\lim_{\lambda \rightarrow \infty }\; \frac{\log \lambda }{\lambda ^4}\;G(\boldsymbol{\lambda },\boldsymbol{Q})= \frac{4}{\pi }[(\boldsymbol{Q}\cdot \hat{\boldsymbol{\lambda }}) + 2]^2.
\end{equation}
But $\|\lamvec_n\|\ra\infty$, so (\ref{limG}) contradicts (\ref{inqEG}) unless $\|\boldsymbol{Q}\|=2$.

Hence   $\| \boldsymbol{Q}\| =2$. Let $\theta $ be the angle between $\boldsymbol{\lambda }$ and $\boldsymbol{Q}$, namely,
\begin{equation}
\boldsymbol{Q}\cdot \hat{\boldsymbol{\lambda }}=\| \boldsymbol{Q}\|  \cos \theta =2 \cos \theta.
\end{equation}
Then, it follows from (\ref{G}) that
\begin{equation}\label{Z}
G(\boldsymbol{\lambda },\boldsymbol{Q})= 4 F_1(\lambda ) \cos^2 \theta  +2 F_2(\lambda ) \cos \theta  +F_3(\lambda )=:Z(\lambda ,\theta ).
\end{equation}
We shall appeal to the following technical lemma whose proof we postpone to an appendix.

\begin{lemma}\label{lem1}
On $(\rat_1, \gamma _{L^2})$, there exist $c_0, \lambda _0 >0 $ such that for all $\lambda \geq \lambda _0$, $Z(\lambda ,\theta )$, given in (\ref{Z}), satisfies
\begin{equation}\label{Z1}
Z(\lambda ,\theta ) > \frac{c_0 \lambda ^4}{(\log \lambda )^3}, \quad \forall\:  \theta \in \mathbb{R}.
\end{equation}
\end{lemma}

\noindent Using the above lemma, it follows from (\ref{inqEG}), (\ref{Z}) and (\ref{Z1}) that, for all $n$ sufficiently large,
\begin{equation}
2 E \geq G(\boldsymbol{\lambda }_n,\boldsymbol{Q}) > \frac{c_0 \|\lamvec_n\| ^4}{(\log \|\lamvec_n\| )^3},
\end{equation}
a contradiction. Hence $\lamvec_n$ is bounded.

Assume now that $\dot\lamvec_n$ is unbounded. We have shown that $\|\lamvec_n\|$ is confined to a closed bounded interval,
so $A_i(\|\lamvec_n\|)$, $B(\|\lamvec_n\|)$ are positive, bounded and bounded away from zero, 
and $\ol{A}(\|\lamvec_n\|)$ is bounded, by continuity. Hence, from (\ref{EE}) we see that
\beq
2E>A_1(\|\lamvec_n\|)\left(1-\frac{\|\lamvec_n\|^2A_1(\|\lamvec_n\|)}{4A_3(\|\lamvec_n\|)}\right)\|\dot\lamvec_n\|^2+c
=A_1(\|\lamvec_n\|)\frac{2\|\lamvec_n\|^2}{1+2\|\lamvec_n\|^2}\|\dot\lamvec_n\|^2+c
\eeq
for some constant $c\in\R$. But this contradicts unboundedness of $\dot\lamvec_n$.

Finally, assume that $\|\Omegavec_n\|$ is unbounded. We have already shown that $\|\lamvec_n\|$ and $\|\dot\lamvec_n\|$ are 
bounded, and by continuity,
$A_i(\|\lamvec_n\|)$ are positive, bounded and bounded away from $0$. But this immediately contradicts (\ref{E}).
\end{proof}

\begin{cor}\label{ratcomp} $(\rat_1,\gamma_{L^2})$ is RMG complete.
\end{cor}

\begin{proof} For each $K>0$, let $X_K=\{({\cal O},\Omegavec,\lamvec,\dot\lamvec\in T\rat_1\: :\: 
\|\Omegavec\|+\|\lamvec\|+\|\dot\lamvec\|)\leq K\}$. By a standard application of Picard's method, there exists $T_K>0$,
depending only on $K$, such that, for all $x_0\in X_K$ there exists a unique RMG curve $x:[-T_K,T_K]\ra X_{2K}$ with
$x(0)=x_0$. Now choose and fix $x_0\in T\rat_1$, and let $X=q^{-1}(q(x_0))$, the level set of $q$ containing $x_0$. By
Theorem \ref{xcomp}, there exists $K>0$ such that $X\subset X_K$. Hence there is a unique RMG curve $x:[-T_K,T_K]\ra X_{2K}$
with $x(0)=x_0$. But, by Proposition \ref{iwtssyc}, $x(\pm T_K)\in X\subset X_K$, so this solution can be extended, both forward
and backward in time, to $[-3T_K,3T_K]$, and $x(\pm 3T_K)\in X\subset X_{K}$ also. Proceeding inductively, the RMG curve has
an extension $x:\R\ra X_K$. Since $x_0$ was arbitrary, it follows that RMG flow is complete.
\end{proof}

\begin{remark} Theorem \ref{xcomp} is strictly stronger than Corollary \ref{ratcomp}, since it implies that every RMG 
curve in $(\rat_1,\gamma_{L^2})$ is confined to a compact subset of $\rat_1$ and hence is bounded away from the
boundary of $\rat_1$ at infinity. This is not true of complete RMG flows in general (consider for example the trivial RMG flow
on $\C^n$). 
\end{remark}

\begin{remark}\label{iblwltrmd}
 Since $(\rat_1,\gamma_{L^2})$ is known to be geodesically incomplete, it is a counterexample to the conjecture
\cite{kruspe-vortex} that every RMG complete k\"ahler manifold is geodesically complete. Simpler counterexamples can be 
constructed. For example the surface of revolution $\C$ given the metric $g={\rm sech}\,|z|\: dz d\bar z$ is manifestly geodesically 
incomplete and can be shown, by an energy/angular momentum conservation argument analogous to the one presented here for 
$\rat_1$, to be RMG complete \cite{alq-thesis}.
\end{remark}

The $L^2$ geometry of $\rat_n$, for $n\geq 2$, is comparatively poorly understood. It is known to be $G$-invariant, k\"ahler
and geodesically incomplete \cite{spe-l2}, and is conjectured to have finite total volume \cite{bap-l2}. Inside $\rat_n$
there is a topologically cylindrical submanifold, $\rat_n^{eq}$, the fixed point set of the 
circle group of isometries $W(z)\mapsto e^{in\alpha}W(e^{-i\alpha}z)$. This consists of rotationally equivariant 
rational maps, of the form
$W(z)=c z^n$, where $c\in\C^{\times}=\C\less\{0\}$, and is preserved by (extrinsic) RMG flow on $\rat_n$ by
Corollary \ref{hatu}. The induced $L^2$ metric on $\rat_n^{eq}$ was studied in detail in \cite{mcgspe}. 
It is interesting to compare the intrinsic RMG flow on $\rat_n^{eq}$ with the extrinsic RMG flow, defined by its inclusion in
$\rat_n$.

Denote by 
$\pi:\C^\times\ra\rat_n^{eq}$ the $n$-fold covering map $a\mapsto [W:z\mapsto (az)^n]$. Then the lifted $L^2$ metric on
$\C^\times$ is
\beq
\pi^*\gamma_{L^2}^{eq}=F(|a|)dad\bar{a},\qquad F(\rho)=\pi n^2\int_0^\infty\frac{s^n}{(1+s^n)^2}\frac{ds}{(\rho^2+s)^2}.
\eeq
Now $c\mapsto c^{-1}$ is an isometry of $\rat_n^{eq}$, whence it follows that $a\mapsto a^{-1}$ is an isometry of the lifted
metric. Furthermore, $\lim_{\rho\ra\infty}F(\rho)$ exists for all $n\geq 2$, so $\pi^*\gamma_{L^2}^{eq}$ has a $C^0$ extension
to $S^2=\C^\times\cup\{0,\infty\}$ for all $n\geq 2$, which we denote $\ol\gamma_n$. For $n$ sufficiently large, we can
obtain useful information about RMG flow on $(\rat_n^{eq},\gamma_{L^2}^{eq})$ by considering its lift to $(S^2,\ol\gamma_n)$. 
This requires us to establish enhanced regularity of $\ol\gamma_n$, as follows.

\begin{prop}\label{gamreg}
For all $n\geq 5$, the extended lifted metric $\ol\gamma_n$ on $S^2$ is $C^3$.
\end{prop}

\begin{proof}
It is known \cite{mcgspe} that $\ol\gamma_n$ is $C^2$ for all $n\geq 4$. Further, $\ol\gamma_n$ is manifestly smooth on
$S^2\less\{0,\infty\}$ so, in light of the isometry $a\mapsto 1/a$, which interchanges $0$ and $\infty$, it suffices to prove
that $f_{xxx},f_{xxy},f_{xyy}$ and $f_{yyy}$ exist at $(0,0)$, where $f(x,y)=F(\sqrt{x^2+y^2})$. By computing in polar coordinates,
$(x,y)=\rho(\cos\theta,\sin\theta)$, one sees that all these third derivatives exist (and vanish) if and only if
\bea
\lim_{\rho\ra 0}\frac1\rho\left(F''(\rho)-\frac{F'(\rho)}{\rho}\right)&=&0\nonumber \\
\label{origirtenski}
\lim_{\rho\ra 0}\left(F'''(\rho)-\frac3\rho\left(F''(\rho)-\frac{F'(\rho)}{\rho}\right)\right)&=&0.
\eea
For each pair of integers $n\geq 2$ and $k\geq 0$, define the function $\eta_{n,k}:(0,\infty)\ra\R$,
\beq
\eta_{n,k}(\rho)=\int_0^\infty\frac{s^n}{(1+s^n)^2}\frac{ds}{(\rho^2+s)^k}.
\eeq
For $0\leq k\leq n$, its integrand is bounded above 
by the integrable function $s^{n-k}/(1+s^n)^2$, so, by the Lebesgue Dominated Convergence Theorem,
\beq\label{ogts}
\lim_{\rho\ra0}\eta_{n,k}(\rho)=\int_0^\infty\frac{s^{n-k}}{(1+s^n)^2}ds<\infty.
\eeq
It follows from the definition of $F$ that
\bea
\frac1\rho\left(F''(\rho)-\frac{F'(\rho)}{\rho}\right)&=&24\pi n^2\rho\eta_{n,4}(\rho)\\
F'''(\rho)-\frac3\rho\left(F''(\rho)-\frac{F'(\rho)}{\rho}\right)&=&-129\pi n^2\rho^3\eta_{n,5}(\rho).
\eea
Hence the required limits (\ref{origirtenski}) follow from (\ref{ogts}) provided $n\geq 5$.
\end{proof}

\begin{cor}\label{fagita} For all $n\geq 5$, the intrinsic RMG flow on $(\rat_n^{eq},\gamma_{L^2}^{eq})$ is incomplete.
\end{cor}

\begin{proof} Assume $n\geq 5$. Then $\ol\gamma_n$ is $C^3$, so its Ricci form is $C^1$. Hence, by standard 
existence and uniqueness theory of ODEs, the RMG flow is globally well-defined on $(S^2,\ol\gamma_n)$. In particular, there
is an RMG curve $a:(-\eps,\eps)\ra S^2$ with $a(0)=0$ and $\dot{a}(0)=1$. Consider the image of $a:(-\eps,0)\ra S^2$ under
the projection $\pi:\C^\times\ra\rat_n^{eq}$. By definition, $\pi$ is a holomorphic isometry, so $\pi\circ a$ is an RMG
curve in $\rat_n^{eq}$, which reaches the singular point $a=0$ in finite time. Hence intrinsic RMG flow in $\rat_n^{eq}$ is
incomplete.
\end{proof}

\begin{remark} By resorting to a case-by-case analysis of the flow on $\rat_n^{eq}$ itself, one can extend the 
conclusion of Corollary \ref{fagita} to all $n\geq 2$ \cite{alq-thesis}. The case $n=2$ is considered below, see Proposition 
\ref{lntsmc}.
\end{remark}

As we have remarked, the extrinsic RMG flow on a totally geodesic complex submanifold of a k\"ahler manifold does not, in
general, coincide with its intrinsic RMG flow, so we cannot conclude from Corollary \ref{fagita} that $(\rat_n,\gamma_{L^2})$
is RMG incomplete for $n\geq 5$: this would follow if the {\em extrinsic} RMG flow on $\rat_n^{eq}$ were incomplete.
Remarkably, although we have little information about the $L^2$ metric on $\rat_2$, 
we have enough to prove that the extrinsic RMG flow on $\rat_2^{eq}$ is {\em complete}. This follows from the following
formula for the restriction of the Ricci form to $\rat_n^{eq}$.

\begin{prop}
\label{lrnaghf}
 Let $\rho\rvert$ be the restriction of the Ricci form $\rho$ of $(\rat_n,\gamma_{L^2})$ to $\rat_n^{eq}$
(that is, $\rho\rvert=\iota^*\rho$  where 
$\iota:\rat_n^{eq}\ra\rat_n$ denotes inclusion) and $\rho^{eq}$ be the intrinsic Ricci form on $(\rat_n^{eq},\gamma_{L^2}^{eq})$. 
Then $\rho\rvert=\d\aa\rvert$ and $\rho^{eq}=\d\aa^{eq}$ where
\ben
\aa\rvert&=&-\left(\sum_{j=0}^{2n}\frac{|\chi|F_j'(\chi)}{2 F_j(\chi)}\right)d\psi,\\
\aa^{eq}&=&-\frac{\chi F_n'(\chi)}{2 F_n(\chi)}\: d\psi,
\een 
and
$$
F_j(\chi)=16\pi\int_0^\infty\frac{s^j}{(1+\chi^2 s^n)^2}\frac{ds}{(1+s)^2}.
$$
The coordinate $\chi e^{i\psi}$ on $\rat_n^{eq}\equiv \C^\times$ corresponds to the rational map $W(z)=\chi e^{i\psi} z^n$.
\end{prop}

\begin{proof} $\rat_n^{eq}$ lies entirely within the coordinate chart on $\rat_n$ on which
\beq
W(z)=\frac{a_0+a_1z+\cdots+a_nz^n}{1+a_{n+1}z+\cdots+a_{2n}z^n}.
\eeq
It is the surface $a_0=\cdots=a_{n-1}=a_{n+1}=\cdots=a_{2n}=0$, $a_n=\chi e^{i\psi}\in\C^\times$. Let 
$G(a_0,\ldots,a_{2n})=\log\det\gamma_{\cdot\bar\cdot}$ where $\gamma_{\cdot\bar\cdot}$ denotes the hermitian matrix of metric
coefficients of $\gamma$ with respect to the local complex coordinates $a_j$. Then \cite[p.\ 82]{bes},
\beq\label{agsmc}
\rho=-i\cd\bar\cd G,
\eeq
so
\beq
\rho\rvert=-i\iota^*\cd\bar\cd G=-i\cd\bar\cd(G\circ\iota)
\eeq
since the inclusion is holomorphic. Now
\beq
\gamma_{j\bar{k}}=16\int_\C \frac{1}{(1+|W(z)|^2)^2}\frac{\cd W}{\cd a_j}\ol{\left(\frac{\cd W}{\cd a_k}\right)}
\frac{dzd\bar{z}}{(1+|z|^2)^2}
\eeq
and
\beq
\left.\frac{\cd W}{\cd a_j}\right|=\left\{\begin{array}{cc}
z^j&0\leq j\leq n\\
-\chi e^{i\psi} z^j&n+1\leq j\leq 2n
\end{array}\right.
\eeq
where the vertical stroke denotes evaluation at the rational map $\chi e^{i\psi}z^n$. It follows that 
$\gamma_{j\bar{k}}\rvert=0$ if $j\neq k$, and that
\beq
\gamma_{j\bar{j}}\rvert=\left\{\begin{array}{cc}
F_j(\chi)&0\leq j\leq n\\
\chi^2 F_j(\chi)&n+1\leq j\leq 2n.
\end{array}\right.
\eeq
Hence
\beq
G\circ\iota=n\log \chi^2+\sum_{j=0}^{2n}\log F_j(\chi),
\eeq
and the formula for $\rho\rvert$ immediately follows. To obtain the formula for $\rho^{eq}$ we note that the
induced metric on $\rat_n^{eq}$ is 
\beq
\gamma_{L^2}^{eq}=\gamma_{n\bar{n}}\rvert da_n d\bar{a}_n=F_n(\chi)dc d\bar{c},
\eeq
where $c=\chi e^{i\psi}$, 
and use (\ref{agsmc}).
\end{proof}

Since the integrand in $F_j$ is rational, one can, in principle, evaluate each of these functions as an explicit
function of $\chi$. The expressions involved become very complicated as $n$ grows large, however. 

Both extrinsic and intrinsic RMG flow on $\rat_n^{eq}$ are governed by a lagrangian of the form
\beq
L=\frac12 F_n(\chi)(\dot\chi^2+\chi^2\dot\psi^2)-a(\chi)\dot\psi
\eeq
where $a(\chi)=\aa^{eq}(\cd/\cd\psi)$ or $a(\chi)=\aa\rvert(\cd/\cd\psi)$ respectively. In each case, both the momentum
conjugate to $\psi$,
\beq
P=\chi^2 F_n(\chi)\dot\psi-a(\chi)
\eeq
and the kinetic energy
\beq
E=\frac12F_n(\chi)(\dot\chi^2+\chi^2\dot\psi^2)=\frac12F_n(\chi)\dot\chi^2+\frac{(P+a(\chi))^2}{2\chi^2 F_n(\chi)}
\eeq
are conserved. This is equivalent to motion on $(0,\infty)$ with the metric $F_n(\chi)d\chi^2$ in the effective
potential
\beq
V_P(\chi)=\frac{(P+a(\chi))^2}{2\chi^2 F_n(\chi)}.
\eeq
Since $(0,\infty)$ has finite total length with respect to this metric \cite{sadspe}, the flow is complete if and only if,
for each $P\in \R$, the effective potential is unbounded above as $\chi\ra 0$ and $\chi\ra\infty$. Both the intrinsic and
extrinsic RMG flows are symmetric under $c=\chi e^{i\psi}\mapsto 1/c$, so in fact it suffices to consider $V_P(\chi)$ in 
a neighbourhood of $0$. 

\begin{prop}\label{lntsmc} 
$\rat_2^{eq}$ is extrinsically RMG complete with respect to the $L^2$ metric, but intrinsically RMG incomplete.
\end{prop}

\begin{proof} As argued above, we must show that the effective potential $V_P$ is unbounded above as $\chi\ra 0$ for all
$P$, in the case of extrinsic flow, and is bounded as $\chi\ra 0$ for at least one choice of $P$ in the case of intrinsic flow.
Let $G_j(\chi)=-\frac12\chi F_j'(\chi)/F_j(\chi)$, and $G(\chi)=\sum_{j=0}^4 G_j(\chi)$. Then the effective potentials 
governing the extrinsic and intrinsic RMG flows are
\beq
V_P^{ext}(\chi)=\frac{(P+G(\chi))^2}{2\chi^2 F_2(\chi)},\qquad
V_P^{int}(\chi)=\frac{(P+G_2(\chi))^2}{2\chi^2 F_2(\chi)},\qquad
\eeq
respectively. With the aid of Maple, for example, one can obtain the following limits:
\beq
\lim_{\chi\ra 0}\chi F_2(\chi)=4\pi,\qquad
\lim_{\chi\ra 0}\frac{G_2(\chi)-\frac12}{\chi\log\chi}=-\frac4\pi,\qquad
\lim_{\chi\ra 0}(G(\chi)-3)\log\chi=-\frac12.
\eeq
It follows that, for all $P\neq -3$,
\beq
\lim_{\chi\ra 0}\chi V^{ext}_P(\chi)=\frac{(P+3)^2}{8\pi}
\eeq
and
\beq
\lim_{\chi\ra 0}\chi(\log\chi)^2 V^{ext}_{-3}(\chi)=\frac{1}{32\pi}.
\eeq
Hence, for all $P$, $V_P^{ext}$ is unbounded above as $\chi\ra 0$. But
\beq
\lim_{\chi\ra 0}V^{int}_{1/2}(\chi)=0
\eeq
so $V^{int}_{1/2}$ is bounded. 
\end{proof}

Numerical analysis of the functions $F_j(\chi)$ suggests that $\rat_n^{eq}$ is likely to be extrinsically RMG complete
for all $n\geq 2$, but we have been unable to prove this so far. Since the process of a single isolated lump collapsing to a
singular spike during RMG flow is prohibited by curvature effects in $\rat_1$, and the same is true for a pair of 
equivariant coincident lumps in $\rat_2$, it is plausible that $\rat_n$ should be RMG complete for all $n\geq 2$, despite 
being geodesically incomplete.

\section{Concluding remarks}
\label{sec:conc}
\news

In this paper we have studied Ricci magnetic geodesic motion on the moduli spaces of abelian Higgs vortices and $\CP^1$
lumps. In so doing we have established that two assertions and one conjecture about this kind of soliton dynamics in the
current literature are false. First, contrary to a claim of Collie and Tong
\cite{colton}, RMG motion on the vortex moduli space
does not coincide with the magnetic geodesic flow proposed earlier by Kim and Lee \cite{kimlee} (and, furthermore, we
have shown that the Kim Lee flow is globally ill-defined). Second, we have shown that, while RMG flow
localizes to fixed point sets of groups of holomorphic isometries, the flow does not, as claimed by one of us
and Krusch \cite{kruspe-vortex}, coincide with the intrinsic RMG flow on the fixed point set. We have seen that on both
the submanifold of centred hyperbolic two-vortices and the space of rotationally equivariant two-lumps, the
intrinsic and extrinsic RMG flows are qualitatively different from one another. 
This aspect of RMG flow is conceptually troubling: since it arises by restricting an infinite dimensional 
dynamical system (a field theory) to a finite dimensional submanifold, it is somewhat strange that further symmetry
reduction is not self-consistent.
Third, we have shown that,
contrary to a conjecture in \cite{kruspe-vortex}, there exist k\"ahler manifolds which are
geodesically incomplete but RMG complete: in fact $(\rat_1,\gamma_{L^2})$
is one such manifold. 

Several interesting open questions remain. Can one, by adapting the methods of Stuart for example \cite{stu-vortex}, prove 
rigorously Collie and Tong's conjecture
that Chern-Simons vortex dynamics is controlled by RMG motion in $M_n$, in the small $\kappa$ (and small energy) 
limit? Or can one rigorously derive some alternative magnetic geodesic flow on $M_n$? Can one develop a point-vortex 
formalism for well-separated Chern-Simons vortices, analogous to the one for standard vortices \cite{spe-static,manspe}? This
would provide formal evidence for, or against, Collie and Tong's conjecture. Treating RMG flow as an interesting dynamical
system on k\"ahler manifolds, can one establish geometric criteria which ensure that RMG completeness implies
geodesic completeness? Can one find examples of RMG complete but geodesically incomplete manifolds with bounded scalar
curvature? Or bounded Ricci curvature? Note that $\rat_1$ and the surface of revolution described in Remark \ref{iblwltrmd}
both have unbounded scalar curvature. 


\subsection*{Appendix: proof of Lemma \ref{lem1}}
\label{appendixA}
\news

\appendix
\renewcommand{\theequation}{A.\arabic{equation}}

 One can obtain using, for example, Maple the following asymptotic formulae for  $F_i(\lambda )$, given in (\ref{FFF}), with respect to the $L^2$ metric on $\rat_1$ as $\lambda \rightarrow \infty $:
\begin{align}
F_1(\lambda )&=\frac{\lambda ^4}{\log \lambda }\;\biggl[ a_1 +\frac{a_2}{\log \lambda } +\frac{a_3}{(\log \lambda )^2}+O\biggl(\frac{1}{(\log \lambda)^3 }\biggr)\biggr],\nonumber\\
F_2(\lambda )&=\frac{\lambda ^4}{\log \lambda }\;\biggl[ b_1 +\frac{b_2}{\log \lambda } +\frac{b_3}{(\log \lambda )^2}+O\biggl(\frac{1}{(\log \lambda)^3 }\biggr)\biggr],\label{asymF}\\
F_3(\lambda )&=\frac{\lambda ^4}{\log \lambda }\;\biggl[ c_1 +\frac{c_2}{\log \lambda } +\frac{c_3}{(\log \lambda )^2}+O\biggl(\frac{1}{(\log \lambda)^3 }\biggr)\biggr],\nonumber
\end{align}
where
\begin{align}
a_1&=\frac{4}{\pi },\quad  & a_2=\frac{2}{\pi } [1-2 \log 2],\quad & a_3=\frac{1}{\pi } [1-4 \log 2 +4 (\log 2)^2],\nonumber\\
b_1&=\frac{16}{\pi },\quad & b_2=\frac{2}{\pi } [3-8 \log 2],\quad & b_3=\frac{1}{\pi } [2-12 \log 2 +16 (\log 2)^2],\\
c_1&=\frac{16}{\pi },\quad & c_2=\frac{4}{\pi } [1-4 \log 2],\quad & c_3=\frac{1}{\pi } [1-32 \log 2 +64 (\log 2)^2].\nonumber
\end{align}
It follows from (\ref{Z}) and (\ref{asymF}) that
\begin{equation}
Z(\lambda ,\theta )=Z_0(\lambda ,\theta )+Z_{\text{error}}(\lambda ,\theta ),
\end{equation}
where
\begin{align}
Z_0(\lambda ,\theta )&=\frac{\lambda ^4}{\log \lambda }\biggl[(4 a_1 \cos^2 \theta  +2 b_1 \cos \theta +c_1)+\frac{1}{\log \lambda } (4 a_2 \cos^2 \theta  +2 b_2 \cos \theta +c_2)\nonumber\\
&\:\:+\frac{1}{(\log \lambda)^2}  (4 a_3 \cos^2 \theta  +2 b_3 \cos \theta +c_3)\biggr],
\end{align}
and $Z_{\text{error}}(\lambda ,\theta )$ satisfies the following estimate: there exist  $c_*, \lambda _* >0$ such that for all $\lambda \geq \lambda _*$,
\begin{equation}
\bigl | Z_{\text{error}}(\lambda ,\theta )\bigr |  < \frac{c_* \lambda ^4}{(\log \lambda )^4}, \quad \forall \; \theta \in \mathbb{R}.
\end{equation}
Hence, it suffices to prove that $Z_0(\lambda ,\theta )$ satisfies an estimate of the form (\ref{Z1}).

Defining $\tau =1+\cos \theta $ and $x=1/\log \lambda $. Then
\begin{equation}\label{ZZZ}
\frac{\log \lambda }{\lambda ^4}\; Z_0(\lambda ,\theta )=P_x(\tau ),
\end{equation}
where 
\begin{equation}
P_x(\tau )=\alpha _1(x) \tau ^2+\alpha _2(x) \tau  +\alpha _3(x),
\end{equation}
and the coefficients $\alpha _1$, $\alpha _2$ and $\alpha _3$ are given by
\begin{align}
\alpha _1(x)&=4(a_1+a_2 x+a_3 x^2),\nonumber\\
\alpha _2(x)&=2(b_2-4 a_2) x+2(b_3-4a_3) x^2,\label{123}\\
\alpha _3(x)&=(4a_3-2b_3+c_3) x^2.\nonumber
\end{align}
Since $\alpha _1(0) >0$, then there exists $x_*>0$ such that for all $x\in (-x_*,x_*)$, $P_x(\tau )$ has a minimum, occurs at $\tau =\tau _*$, where $dP_x(\tau )/d\tau \Bigl\lvert_{\tau =\tau _*}=0$, that is,
\begin{equation}
\tau _*(x)=-\frac{1}{2}\frac{\alpha _2(x)}{ \alpha _1(x)}.
\end{equation}
So,  for all $x\in (-x_*,x_*)$, the minimum value of $P_x(\tau )$ is
\begin{equation}\label{Pxtstar}
P_x(\tau _*(x))=-\frac{1}{4 \alpha _1(x)} [\alpha _2(x)^2-4 \alpha _1(x) \alpha _3(x)] .
\end{equation}
Note that $P_x(\tau _*(x))$ is a rational function of $x$, and hence is analytic. Using (\ref{123}), one  finds that
\begin{equation}
P_0(\tau _*(0))=0,\qquad \frac{d}{dx}P_x(\tau _*(x))\Bigl\lvert _{x=0} =0,
\end{equation}
and
\begin{equation}
\frac{d^2}{dx^2} P_x(\tau _*(x))\Bigl\lvert _{x=0}=-\frac{1}{8 a_1}[ (b_2-4 a_2)^2 -16 a_1 (4 a_3-2b_3 +c_3)] >0.
\end{equation}
Thus, there exist $ \varepsilon>0$ and   $0<x_0 <x_*$ such that for all $x\in (-x_0,x_0)$,
\begin{equation}
P_x(\tau _*(x)) \geq  \varepsilon \; x^2.
\end{equation}
Hence, for all $x\in (0,x_0)$,
\begin{equation}
P_x(\tau ) \geq \varepsilon  \; x^2,\qquad \forall \;\tau \in \mathbb{R}.
\end{equation}
Hence, it follows from (\ref{ZZZ}) that for all $\lambda > e^{1/x_0}$, 
\begin{equation}
\frac{\log \lambda }{\lambda ^4}\; Z_0(\lambda , \theta )= P_x(\tau )  \geq \varepsilon \; x^2 = \frac{\varepsilon }{(\log \lambda)^2 },\quad \forall \:\theta \in \mathbb{R},
\end{equation}
which implies that $Z_0(\lambda , \theta )$ satisfies the estimate (\ref{Z1}).
\hfill $\Box$


\subsection*{Acknowledgements}

 The work of JMS was supported by the UK
Engineering and Physical Sciences Research Council, and that of
LSA by a scholarship from King Abdulaziz University (Jeddah, KSA).

\end{document}